\documentclass[runningheads]{llncs}
\usepackage[T1]{fontenc}
\usepackage[font=small,figurename=Figure.]{caption}
\usepackage{algcompatible}
\usepackage{algpseudocode, algorithm}
\usepackage{multicol}
\usepackage{multirow}
\usepackage{float}
\usepackage{graphicx}
\usepackage{amsmath, xparse}
\usepackage{booktabs}
\usepackage{todonotes}
\usepackage{pgfplots}
\pgfplotsset{compat=1.14}
\usepackage{xspace}
\usepackage{caption}
\usepackage{appendix}
\usepackage{subcaption}
\newcommand{\punt}[1]{}
\newcommand{\cmnt}[1]{}
\newcommand{\ignore}[1]{}
\newcommand{\remove}[1]{}

\newcommand{\figref}[1]{Fig.~\ref{fig:#1}}

\newcommand{\algoref}[1]{{Algorithm \ref{alg:#1}}}

\newcommand{\Figref}[1]{Figure~\ref{fig:#1}}

\newcommand{\Apnref}[1]{Section~\ref{apn:#1}}

\newcommand{\sct} {SCT\xspace}

\begin{document}
\title{Efficient and Concurrent Transaction Execution Module for Blockchains}
%
%
\author{Manaswini Piduguralla \inst{1} \and
Saheli Chakraborty\inst{1} \and
Parwat Singh Anjana \inst{1} \and
Sathya Peri \inst{1} }
\authorrunning{M. Piduguralla et al.}
%
\institute{Indian Institute of Technology Hyderabad, Kandi, Telangana 502285
\email{\{cs20resch11007,ai20mtech12002,cs17resch11004,sathya\_p\}@iith.ac.com}
}
\maketitle              
\begin{abstract}
Blockchain technology is a distributed, decentralized, and immutable ledger system. It is the platform of choice for managing smart contract transactions (SCTs). Smart contracts are self-executing codes of agreement between interested parties commonly implemented using blockchains. A block contains a set of transactions representing changes to the system and a hash of the previous block. The SCTs are executed multiple times during the block production and validation phases across the network. Transaction execution is sequential in most blockchain technologies.

In this work, we propose a parallel direct acyclic graph (DAG) based parallel scheduler module for concurrent execution of SCTs. Using the module, which can be seamlessly integrated into the blockchain framework, parallel transactions can be executed more efficiently and the speed of transactions can be increased. The dependencies among a block's transactions are represented through a concurrent DAG data structure that assists with throughput optimization. We have created a DAG scheduler module that can be incorporated into blockchain platforms for concurrent execution with ease. For evaluation, our framework is implemented in Hyperledger Sawtooth V1.2.6. The performance across multiple smart contract applications is measured for various scheduler types. A performance comparison between the proposed framework and the Sawtooth parallel SCT execution framework shows significant improvements.

\end{abstract}
\section{Introduction}
\label{apn:intro}
Blockchain platforms help establish and maintain a decentralized and distributed ledger system between untrusting parties \cite{Nakamoto:Bitcoin:2009}. The blockchain is a collection of immutable blocks typically in the form of a chain. Each block points to its previous block by storing its hash. A block in the blockchain consists of several \emph{smart contract transactions (\sct{s})} which are self-executing contracts of agreement between two or more parties that are written in the form of computer code. These help in the execution of agreements among untrusted parties without the necessity for a common trusted authority to oversee the execution. The development and deployment of smart contracts on blockchain platforms is growing rapidly.

A blockchain network usually consists of several nodes (ranging from thousands to millions depending on the blockchain), each of which stores the entire blockchain's contents. Any node in the blockchain can act as a \emph{block producer}. A producer node selects transactions from a pool of available transactions and packages them into a block. The proposed block is then broadcast to other nodes in the network. A node on receiving the block acts as a \emph{validator}. It validates the transactions in the block by executing them one after another. Thus a node can act as a producer while producing the block and as a validator while validating the block. 

Agreement on the proposed block by the nodes of the blockchain is performed through various consensus mechanisms, like proof of work (PoW)\cite{Nakamoto:Bitcoin:2009}, proof of stake (PoS)\cite{Vasin:2014:WP}, proof of elapsed time (PoET)\cite{KunBla+:1997:CJ}, etc. For a block to be added to the blockchain, the transactions of the block are processed in two contexts: (a) first time by the block producer when the block is produced; (b) then by all the block validators as a part of the block validation. Thus the \sct code is executed multiple times by the producer and the validators. 

The majority of blockchain technologies execute the \sct{s} in a block serially during the block creation and validation phases. This is one of the bottlenecks for higher throughput and scalability of blockchain models \cite{Dickerson+:ACSC:PODC:2017}. Throughput of the blockchain can be improved by concurrent execution of transactions. To process transactions simultaneously, it is necessary to make sure that transactions that are dependent on each other are not run together. In addition, concurrent execution of transactions at each validator should result in the same end state in the database. In this work, a framework for the concurrent execution of transactions on producers and validators is proposed. The framework has been implemented in Hyperledger Sawtooth 1.2.6 to determine the level of optimization it achieves.

The underlying architecture of Hyperledger Sawtooth and its introduction are provided in \Apnref{background}. \Apnref{framework} details our framework for parallel transaction execution as well as pseudocode for DAG creation, utilization, and smart validation. This is followed by a discussion of the implementation details, experiment design, and results presented in \Apnref{experiment}. Tests of our scheduler have been extensively conducted, and the results of our experiments have been analyzed in detail. We discuss our conclusions and next steps in \Apnref{conclusion} of our paper.
\section{Background}
\label{apn:background}
\subsection{Hyperledger Sawtooth}

\label{subsec:sawtooth}

The Hyperledger Foundation is an open-source collaboration project by the Linux Foundation to establish and encourage cross-industry blockchain technologies. Sawtooth is one of the most popular blockchain technologies being developed for permission and permissionless networks. It is designed such that transaction rules, permissions, and consensus algorithms can be customized according to the particular area of application. Some of the distinctive features of Sawtooth are modularity, multi-language support, parallel transaction execution, and pluggable consensus. The modular structure of Sawtooth helps in modifying particular operations without needing to make changes throughout the architecture.

\begin{figure}[]

\centering

\includegraphics[scale=0.25]{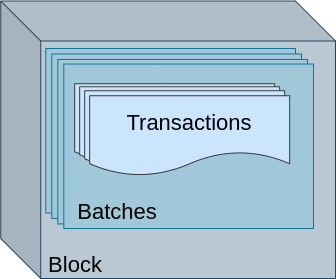}

\caption{Structure of a Hyperledger Sawtooth block}

\label{fig:block}

\end{figure}

In Sawtooth, smart contracts are referred to as transaction families, and the logic for the contract is present in the respective families' transaction processors. Modifications to the state are performed through transactions and they are always wrapped inside a batch. A batch is the atomic unit of change in the system and multiple batches are combined to form a block (\Figref{block}). The node architecture of Sawtooth includes five modules that play crucial roles in blockchain development: global state, journal, transaction scheduler, transaction executor, and transaction processor. The global state containing the data of transaction families of Sawtooth is stored using a radix Merkel tree data structure. The Journal module contains a block completer, block validator and block publishers that deal with creating, verifying and publishing blocks. It is the Transaction Scheduler module's responsibility to identify the dependencies between transactions and schedule transactions that result in conflict-free execution. In order to execute a transaction, the transaction executor collects the context of the transaction. The detailed architecture is expained in appendix A \cite{piduguralla+:2023:arxiv}.

Hyperledger Sawtooth architecture includes a parallel transaction scheduler (tree-scheduler) that uses a Merkle-Radix tree with nodes that are addressable by state addresses. This tree is called the predecessor tree. Each node in the tree represents one address in the Sawtooth state and a read list and write list are maintained at each node. Whenever an executor thread requests the next transaction, the scheduler inspects the list of unscheduled transactions and the status of their predecessors. The drawbacks of the tree scheduler are that the status of predecessor transactions needs to be checked before a transaction can be scheduled. The construction of the tree data structure is serial. The number of addresses accessed in a block is generally higher than the total number of transactions. The memory space utilized for an address based data structure is more than a transaction based data structure. The miners and validators both construct the tree at their end instead of sharing the tree constructed by the miners.

The proposed framework for transaction execution on the blockchain would improve the throughput of SCTs by making the block creation and validation process concurrent. SCTs that are independent of each other are executed in parallel in the framework. The dependencies are represented as a direct acyclic graph (DAG) based on transaction inputs and outputs. DAG sharing and smart validator module designs are also included in the framework to further optimize block validation.

\section{Proposed Framework}
\label{apn:framework}
In this section, the proposed framework for parallel transaction execution in blockchains through static analysis by miners and validators of the network is detailed. This framework introduces \emph{parallel scheduler} and \emph{smart validator} modules into the blockchain node architecture. The parallel scheduler is responsible for identifying the dependencies among the transactions in the block and scheduling them for conflict-free execution. This is done by determining the dependencies among the transactions. The identified dependencies are represented by a direct acyclic graph (DAG) that is shared along with the block to minimize the validation time of the blockchain. DAG shared along with the blocks are validated by the \emph{smart validator}, which determines if the miners have shared incorrect graph. In this section, a detailed framework \Figref{framework} and the pseudo-codes of the modules are presented. The implementation and experiments conducted in Hyperledger Sawtooth are discussed in further sections.

\begin{figure}[h]
    \centering
	{\includegraphics[scale=0.23]{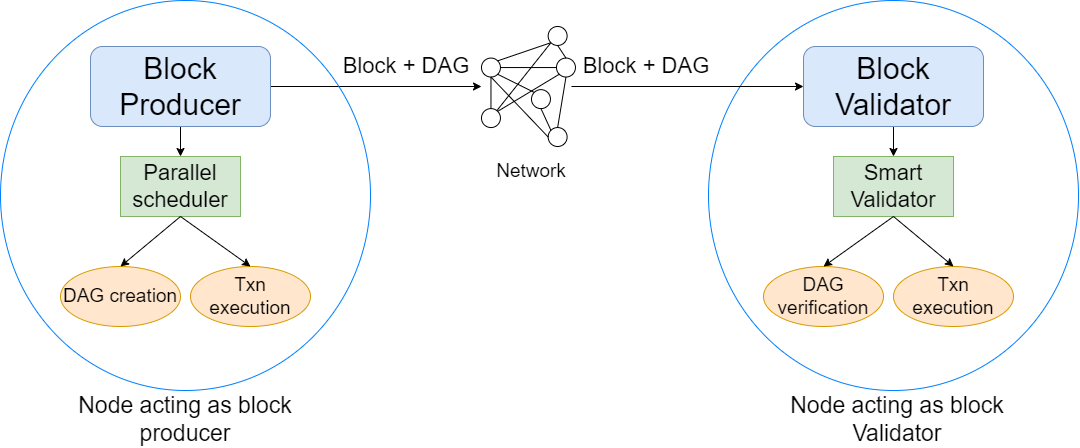}}
	 \caption{Proposed framework in the blockchain}
	\label{fig:framework}
\end{figure}

\subsection{Parallel Scheduler}
The parallel scheduler module is part of the miner in the proposed framework. It performs the operations of DAG creation and conflict-free transaction execution. Both operations are multi-threaded for higher throughput. 

\noindent
\textbf{DAG creation:}
The DAG creation is initiated when the miner creates the next block in the blockchain. In blockchains like Sawtooth, the addresses that the transactions read from, and write to, are present in the transaction header. Using this information, we can derive the addresses based on the transaction details without having to execute the transaction. Analyzing these input (read) and output (write) addresses, the parallel scheduler computes the dependencies between transactions.

On receiving the block from the block producer, transaction ID is assigned to the transactions based on the position of them in the block ($T_1, T_2, T_3$...). Multiple threads are deployed to generate the DAG. Each thread in the parallel scheduler claims the next unclaimed transaction in the block and a global atomic counter is maintained to keep track. The input addresses of the transaction ($T_a$) are compared with all the output addresses of transactions (e.g., $T_b$) with a higher ID. Correspondingly, the output addresses of $T_a$ are compared with the input and output addresses of $T_b$. If there are any common addresses identified in the above checks, an edge is added from $T_a$ to $T_b$ in the DAG. An adjacency matrix data structure is implemented for representing the graph, and an atomic array to maintain the indegree count of the vertices. The linked list DAG graph are illustrated in the \figref{DAG}. We have also implemented an parallel scheduler module with concurrent linked list structure Fig.~\ref{fig:LLDAG}. The pseudo-code is detailed in \algoref{createDAG}.
\begin{figure}[!tbp]
  \centering
  \begin{minipage}[b]{0.4\textwidth}
   { \includegraphics[scale=0.25]{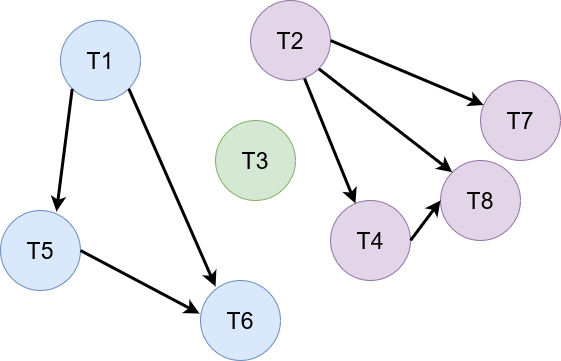}
	 \caption{Direct acyclic graph representation of the dependencies in the block.}
	\label{fig:DAG}
 }
  \end{minipage}
  \hfill
  \begin{minipage}[b]{0.48\textwidth}
   { \centering 
   \includegraphics[scale=0.3]{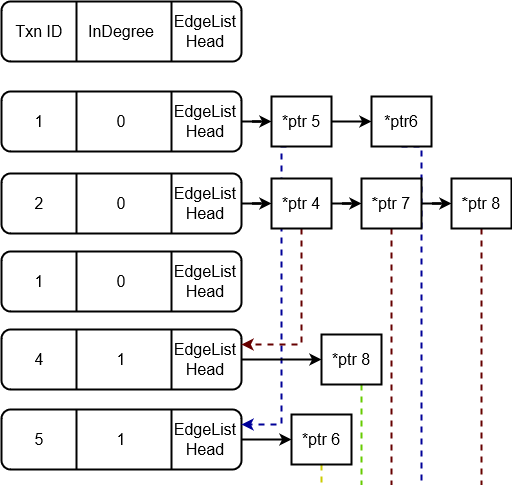}
	 \caption{Linked list representation of the DAG.}
	\label{fig:LLDAG}}
  \end{minipage}
\end{figure}
\begin{algorithm}
    
     { \scriptsize 
    \caption{Multi-threaded createDAG(): $m$ threads concurrently create the DAG}
    \label{alg:createDAG}
    \begin{algorithmic}[1]
    \Procedure{createDAG}{$block$}\Comment{ The block to be produced or validated is the input}
       \While{$1$}
          \State $T_{i} \gets txnCounter.get\&inc()$ \Comment{Claim the next transaction available}
          \If{$T_{i} > txn\_count$}
          \State $txnCounter.dec()$ 
          \State return \Comment{Return if all the transactions are processed}
          \EndIf
          \State Graph\_Node $*txn$     = new Graph\_Node
          \State DAG$\rightarrow$add\_node($T_{i}$, $txn$)\Comment{adding the node to the graph}
          \For{$T_{j}=T_{i}+1$ to $ txn\_count$}\Comment{finding dependent transactions}
            \State $flagEdge$=false
            \If{$T_{i}.readList \cap T_{j}.writeList$}\Comment{Checking for RW and WW conflicts}
            \State $flagEdge$=True
            \EndIf
            \If{$T_{i}.writeList \cap T_{j}.readList$}
            \State $flagEdge$=True
            \EndIf
            \If{$T_{i}.writeList \cap T_{j}.writeList$}
            \State $flagEdge$=True
            \EndIf
            \If{$flagEdge$}
            \State DAG$\rightarrow$add\_edge($T_{i}$, $T_{j}$)
            \EndIf
          \EndFor
       \EndWhile
    \EndProcedure \Comment{The threads end when the DAG is complete} 
    \end{algorithmic}
}
\end{algorithm}

\vspace{1mm}
\noindent
\textbf{Transaction execution:}
For parallel transaction execution, we initiate multiple threads to process the transactions in the block. Each thread selects a transaction for execution using the indegree array. If the indegree of a transaction is zero this indicates that the transaction does not have any predecessor dependent transactions and can be executed ($T_1, T_3,$ and $T_2$ in \figref{DAG}). If no such transactions are available, the threads wait until one is available or end execution if all the transactions have completed execution. Upon selecting a transaction, it is executed and the indegree of all the outgoing edge transactions ($T_5$ and $T_6$ for $T_1$) is decremented by 1. The search for a transaction with indegree zero can be further optimized by initiating the search from the last transaction ID selected. This further reduces the time it takes to find the next transaction. The pseudo-code that multiple threads execute while selecting a transaction is present in \algoref{selectTxn}
\begin{algorithm}
{
\scriptsize
\caption{Multi-threaded selectTxn(): threads concurrently search the DAG for next transaction to execute}
\label{alg:selectTxn}
\begin{algorithmic}[1]
\makeatletter\setcounter{ALG@line}{26}\makeatother
\Procedure{selectTxn}{$DAG$}
      \For{$T_{i}=pos$ To $txn\_count$}\Comment{iterate over until all transactions to find transaction for execution}
        \If{$T_{i}.indeg == 0$}\Comment{Checking for txn with zero indegree}
        \If{$T_{i}.indeg.CAS(0,-1)$}
        \State $pos \gets T_{i}$ \Comment{store the last position for fast parsing}
        \State return $T_{i}$
        \EndIf
        \EndIf
        \EndFor
        \For{$T_{i}=0$ To $pos$}\Comment{iterate over until all transactions to find transaction for execution}
        \If{$T_{i}.indeg == 0$}\Comment{Checking for txn with zero indegree}
        \If{$T_{i}.indeg.CAS(0,-1)$}
        \State $pos \gets T_{i}$ \Comment{store the last position for fast parsing}
        \State return $T_{i}$
        \EndIf
        \EndIf
        \EndFor
        \State return $-1$ \Comment{The threads end when a transaction is selected or all the transaction are executed} 
\EndProcedure 
\end{algorithmic}
}
\end{algorithm}

\subsection{Smart Validatior}

\textbf{DAG sharing:} 
In this phase, the DAG created by the block creator in the blockchain network is shared with the validators in the network. In this way, the blockchain network's validators can conserve the time taken for DAG creation. Malicious miners in the network can exploit the DAG sharing feature to disrupt blockchain functioning. Sharing an inapt DAG with one or more edges missing could lead to validators malfunctioning. A DAG with additional edges will increase the validation process and slow the whole blockchain network.
\vspace{1mm}

\noindent
\textbf{DAG validation:} 
The DAG created by the block creator in the blockchain network is shared with the validators in the network. In this way, the blockchain network's validators can conserve the time taken for DAG creation. In order to address this security risk that is present in DAG sharing, a smart validator is designed for verifying DAGs. The smart validator checks for missing edges or extra edges present in the DAG shared. This is performed by multiple threads for swift graph verification. For all the addresses accessed in the block, a read list and a write list are maintained. By parsing the transactions in the block, transaction IDs are added to the read and write lists of respective addresses. First the check for missing edges is performed by making sure that the transactions in the write list of an address have an edge with all the transactions present in the respective read and write lists. A failed check indicates that the DAG shared has a missing edge. During the check, the number of outgoing edges is calculated for each transaction. We compare the sum of the outgoing edges obtained in the above operation with the in-degree array shared along the block. This function identifies if any extra edges are present in the DAG. As a result of this process, the smart validator verifies the DAG information and recognizes malicious miners. The procedure to handle such miners depends on the type and functionalities of blockchain technology. Malicious miners in the network can exploit the DAG sharing feature to disrupt blockchain functioning. Sharing an inapt DAG with one or more edges missing could lead to validators malfunctioning. A DAG with additional edges will increase the validation process and slow the whole blockchain network.
\begin{algorithm}[h]
{\scriptsize
\label{Algo:smartValidator}
\caption{Multi-threaded smartValidator(): $m$ threads concurrently verify the DAG shared}
\begin{algorithmic}[1]
\makeatletter\setcounter{ALG@line}{45}\makeatother
\Procedure{smartValidator}{$DAG$}
   \While{$ !mMiner$}
      \State $Adds \gets addsCounter.get\&inc()$ \Comment{Claim the next address for analyzing}
      \If{$Adds > adds\_count$}
      \State $addsCounter.dec()$ 
      \State return \Comment{Return if all the address are processed}
      \EndIf
      \For{$i=0$ To $lenght(Adds.readList)$}\Comment{procedure for checking for missing edges}
        \State $T_{i} \gets Adds.readList[i]$
        \For{ $j=0$ To $lenght(Adds.writeList)$}\Comment{read-write dependencies}
        \State $T_{j} \gets Adds.writeList[j]$
        \If{ !checkEdge($T_{i},T_{j}$)}
        \State $mMiner \gets True$
        \State return
        \EndIf
    \State incDeg($T_{i},T_{j}$) \Comment{Increment the indegree of lower txn and mark the edge}
        \EndFor
        \EndFor
        \For{$i=0$ To $lenght(Adds.writeList)$}
        \State $T_{i} \gets Adds.writeList[i]$
        \For{ $j=0$ To $lenght(Adds.writeList)$}\Comment{write-write dependencies}
        \State $T_{j} \gets Adds.writeList[j]$
        \If{ !checkEdge($T_{i},T_{j}$)}
        \State $mMiner \gets True$
        \State return
        \EndIf
    \State incDeg($T_{i},T_{j}$) \Comment{Increment the indegree of lower txn and mark the edge}
        \EndFor
        \EndFor
        \For{ $i=0$ To $ txn\_count$}\Comment{procedure for checking for extra edges}
        \If{ $T_{i}.inDeg !=  T_{i}.calDeg$ } \Comment{Check if shared indegree is equal to calculated indegree}
        \State $mMiner \gets True$
        \State return
        \EndIf
        \EndFor
   \EndWhile
\EndProcedure
\end{algorithmic}
}
\end{algorithm}
\section{Experiments Analysis} 
\label{apn:experiment}

\subsection{Implementation details}
We have chosen Hyperledger Sawtooth as our testing platform since it has a parallel scheduler. To incorporate the DAG framework into the Sawtooth architecture, we have to modify the current parallel scheduler module. Due to the modular nature of Sawtooth, any modifications made to a module can be restricted within the module itself without impacting the whole architecture. For this we need to make sure the functions through which other modules interact with the parallel scheduler module are modified with care.

We have now implemented the DAG sharing and Smart validator modules in Sawtooth 1.2.6. Our modulues are in CPP language while the Sawtooth is a combinatin of Rust and Python. We have choosen CPP for its concurrent capabilities with low overhead cost. For DAG sharing, we have modified the block after the block creator has verified that all the transactions in the block are valid. In Sawtooth 1.2.6 we used the input and output addresses present in the transaction structure. Every transaction in the DAG is represented by a graph node and the outgoing edges indicate dependent transactions. In order to ensure smart validation, the DAG is shared across the block. We have used the dependencies list component of the transaction structure to incorporate DAG into the block. This modified block is shared with the network for validation. 

Initially, we implemented the DAG using a linked list data structure. This is ideal when the size of the graph is unknown and the graph needs to be dynamic. Given that the number of transactions in a block does not change and there is a limit to the number of transactions a block can contain, we have designed an adjacency matrix implementation for DAG. The results have shown further improvement over the linked list implementation. This is because the adjacency matrix is direct access whereas the linked list implementation would require parsing the list to reach the desired node.

In Sawtooth 1.2.6 block validators, the smart validator is implemented. The DAG is recreated using the dependencies list provided in the transaction by the block creator. This optimizes the time taken to create the DAG again in the validators by comparing all the input and output addresses in the transactions. Once created the smart validator first checks for missing edges that should have been present in the DAG shared by the block creator. This is done by generating a list of transactions that access each address and verifying if an edge is present between them. While doing this, we all keep track of the outgoing edge count for each transaction. If the count matches the DAG shared by the block creator that indicates no edges are missing. This way we verify the DAG shared faster than creating the DAG again, optimizing the throughput.
\vspace{1mm}

\noindent
\textbf{Transaction Families: }
We have implemented four transaction families SimpleWallet, Intkey, Voting and Insurance. In SimpleWallet one can create accounts, deposit, withdraw and transfer money. In Intkey clients are increment and decrement values stored in different variables. In Voting transaction family the transactions are create parties, add voters and voters can vote to one of the parties. Insurance family is a data storage transaction family where user details like ID, name, address details are stored and manipulated. To control the percentage of conflicts between transactions, one must have control over the keys created. We have modified the batch creation technique in these transaction families to allow the user to submit multiple transactions in a batch. This way we can not just control the number of transactions in batch but also the conflicts among the transactions in a batch. We individually observed each transaction family behaviour under various experiments and a mix of all four types of types of transaction in a block.
\subsection{Experiments}
We have extensively tested our proposed framework by implementing it in Hyperledger Sawtooth 1.2.6. To test the framework across different scenarios we have devised three conflict parameters (CP) that indicate the level of dependency among the transaction. Conflict parameter one (CP1) is the portion of transactions that have at least one dependency. CP2 ratio of dependencies (edges) to total possibles dependencies. CP3 degree of parallelism in the DAG i.e., the number of independent components in the graph. We have designed four experiments, each varying one parameter while the rest of the parameters are constant. The four parameters are number of blocks, number of transaction in the block, degree of dependency and number of threads. Each named experiment setup one to four respectively. We have named adjacency matrix implementation of our proposed framework as \emph{Adj\_DAG}, linked list implementation as \emph{LL\_DAG}. The Sawtooth inbuilt parallel scheduler used tree data structure, accordingly we have named it as \emph{Tree} and serial execution output as \emph{Serial} in our results. We have observed that due to the presence of global lock in python, the change in number of threads has not impacted the performance significantly. Due to this we have not presented the experiment four results in this work but are available in appendix C \cite{piduguralla+:2023:arxiv}.

It can be observed from \Figref{combined_plots}, that adjacency matrix and linked list implementation of our proposed framework perform significantly better than the tree based parallel scheduler in Sawtooth and serial execution. We have illustrated here some of the experiments we have conducted and the rest can be found in appendix C \cite{piduguralla+:2023:arxiv}. Figure~\ref{fig:bank_E1_CP1},~\ref{fig:voting_E1_CP3} and~\ref{fig:Mixed_E1_CP3} illustrate the impact of change in number of blocks on various schedulers. On average the speedup of \emph{Adj\_DAG} over \emph{Serial} is $1.58$ and \emph{LL\_DAG} is $1.43$ times, while \emph{Tree} is $1.22$. The average speedup of  \emph{Adj\_DAG} over \emph{Tree} is $1.29$ and \emph{LL\_DAG} is $1.17$ times.

\begin{figure*}
    \vspace{-0.5cm}
    \centering
    \begin{subfigure}[b]{0.32\textwidth}
        \centering
        \includegraphics[width=\textwidth]{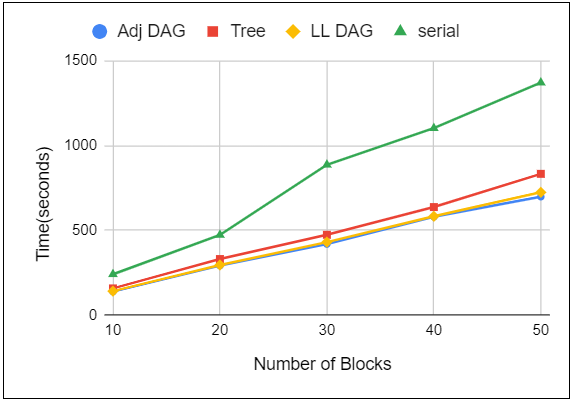}
        \vspace{-1.4\baselineskip}
        \caption{\tiny SimpleWallet: Experiment one, CP1}
        \label{fig:bank_E1_CP1}
    \end{subfigure}
    \begin{subfigure}[b]{0.32\textwidth}
        \centering
        \includegraphics[width=\textwidth]{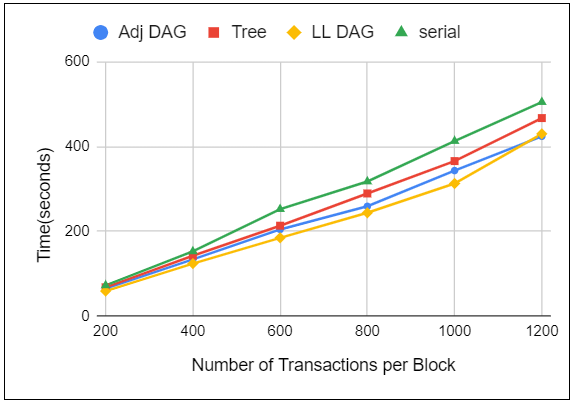}
        \vspace{-1.4\baselineskip}
        \caption{\tiny SimpleWallet: Experiment two, CP2}
        \label{fig:bank_E2_CP2}
    \end{subfigure}
    \begin{subfigure}[b]{0.32\textwidth}
        \centering
        \includegraphics[width=\textwidth]{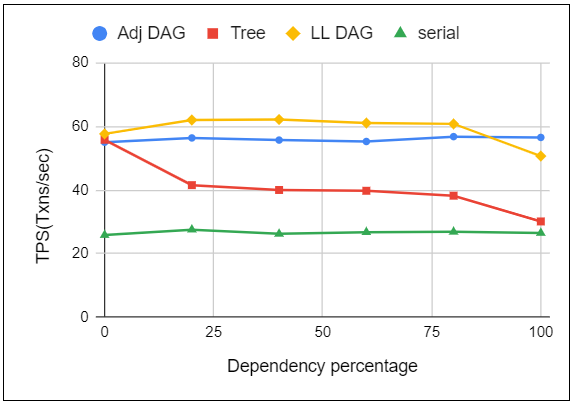}
        \vspace{-1.4\baselineskip}
        \caption{\tiny SimpleWallet: Experiment three, CP3}
        \label{fig:bank_E3_CP3}
    \end{subfigure}
    \begin{subfigure}[b]{0.32\textwidth}
        \centering
        \includegraphics[width=\textwidth]{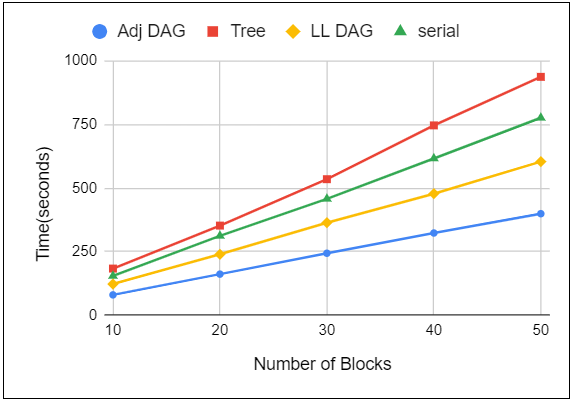}
        \vspace{-1.4\baselineskip}
        \caption{\tiny Voting: Experiment one, CP3}
        \label{fig:voting_E1_CP3}
    \end{subfigure}
    \begin{subfigure}[b]{0.32\textwidth}
        \centering
        \includegraphics[width=\textwidth]{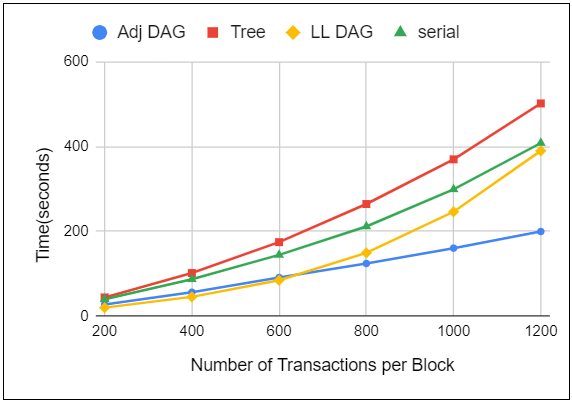}
        \vspace{-1.4\baselineskip}
        \caption{\tiny Voting: Experiment two, CP2}
        \label{fig:voting_E2_CP2}
    \end{subfigure}
    \begin{subfigure}[b]{0.32\textwidth}
        \centering
        \includegraphics[width=\textwidth]{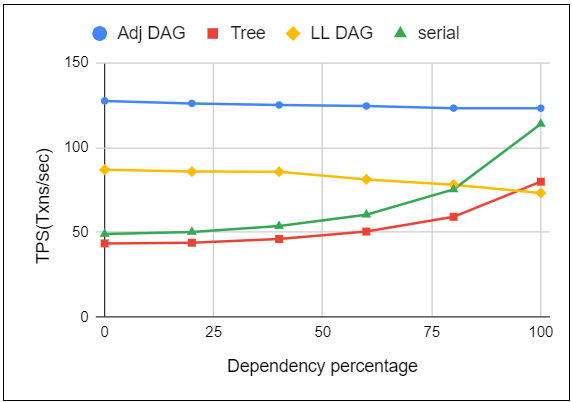}
        \vspace{-1.4\baselineskip}
        \caption{\tiny Voting: Experiment three, CP1}
        \label{fig:voting_E3_CP1}
    \end{subfigure}
    \begin{subfigure}[b]{0.32\textwidth}
        \centering
        \includegraphics[width=\textwidth]{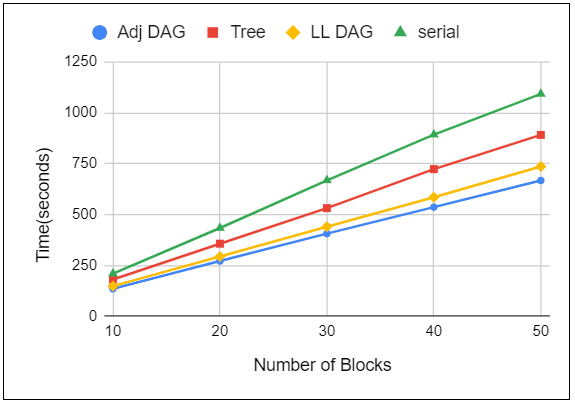}
        \vspace{-1.4\baselineskip}
        \caption{\tiny Mixed Block: Experiment one, CP3}
        \label{fig:Mixed_E1_CP3}
    \end{subfigure}
    \begin{subfigure}[b]{0.32\textwidth}
        \centering
        \includegraphics[width=\textwidth]{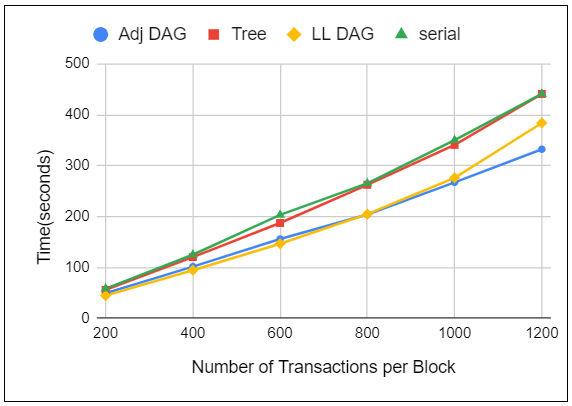}
        \vspace{-1.4\baselineskip}
        \caption{\tiny Mixed Block: Experiment two, CP2}
        \label{fig:Mixed_E2_CP2}
    \end{subfigure}
    \begin{subfigure}[b]{0.32\textwidth}
        \centering
        \includegraphics[width=\textwidth]{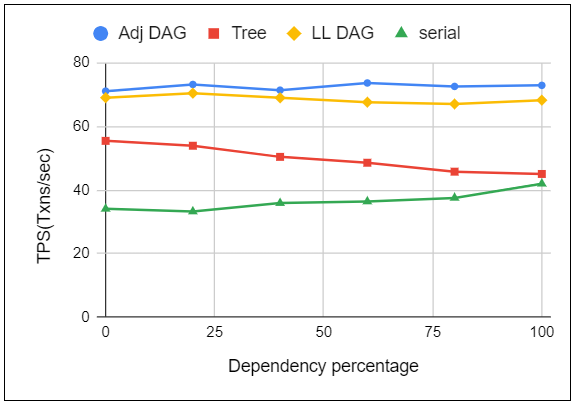}
        \vspace{-1.4\baselineskip}
        \caption{\tiny Mixed Block: Experiment three, CP1}
        \label{fig:Mixed_E3_CP1}
    \end{subfigure}
    \begin{subfigure}[b]{0.32\textwidth}
        \centering
        \includegraphics[width=\textwidth]{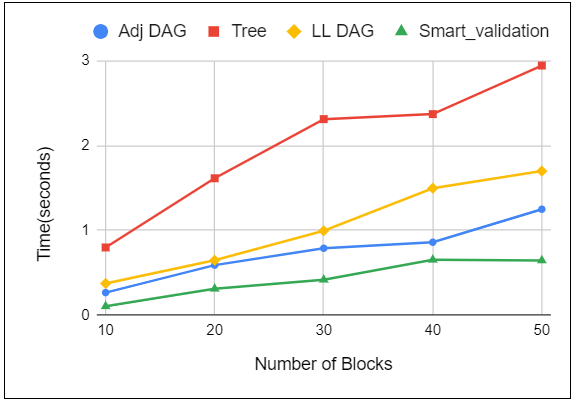}
        \vspace{-1.4\baselineskip}
        \caption{DS creation time varying the number of blocks}
        \label{fig:Data_structure_1}
    \end{subfigure}
    \begin{subfigure}[b]{0.32\textwidth}
        \centering
        \includegraphics[width=\textwidth]{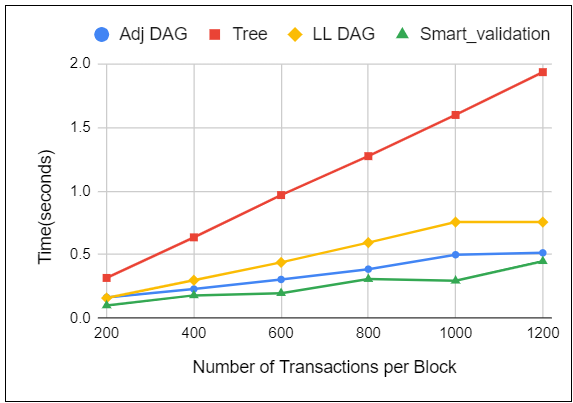}
        \vspace{-1.4\baselineskip}
        \caption{DS creation time varying the number of transaction per Block}
        \label{fig:Data_structure_2}
    \end{subfigure}
    \begin{subfigure}[b]{0.32\textwidth}
        \centering
        \includegraphics[width=\textwidth]{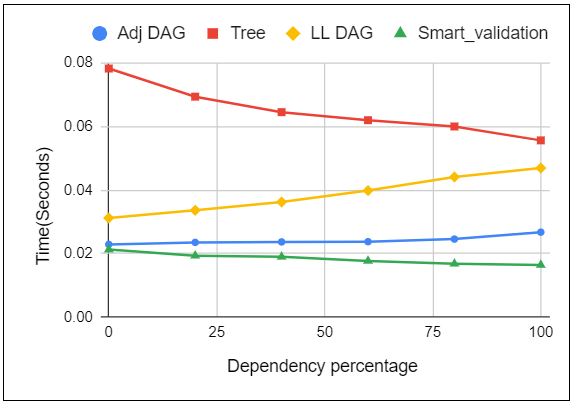}
        \vspace{-1.4\baselineskip}
        \caption{DS creation time varying the dependency percentage}
        \label{fig:Data_structure_3}
    \end{subfigure}
    \vspace{-1\baselineskip}
    \caption{Detailed analysis of our proposed framework performance with both adjacency matrix and linked list implementations in Hyperledger Sawtooth 1.2.6.}
    \label{fig:combined_plots}
\end{figure*}

Experiment two results are depicted in Figures~\ref{fig:bank_E2_CP2},~\ref{fig:voting_E2_CP2} and ~\ref{fig:Mixed_E2_CP2}. We can observe that the gap between serial and the parallel schedulers increases with increase in number of transactions in block. As higher number of transactions more scope for concurrency. In Experiment three we have varied the degree of dependency between the transactions to over its impact on the transactions per second (TPS). The dependency amongs the tranactions is increased by making multiple transactions access the same accounts/addresses. Ideally varying the conflict percentage without changing the number of transactions should not impact serial execution throughput. But decrease in number of memory accesses improves the execution time. We can observe this phenomenon in serial execution time in Figures~\ref{fig:bank_E3_CP3},~\ref{fig:voting_E3_CP1} and~\ref{fig:Mixed_E3_CP1}. We can also observe that while in \emph{ADJ\_DAG} and \emph{LL\_DAG} the effects balance and we can observe a steady TPS, in \emph{Tree} scheduler the performance further decreases the TPS.

Figures~\ref{fig:voting_E1_CP3},~\ref{fig:voting_E2_CP2} and ~\ref{fig:voting_E3_CP1} show the Voting transaction family behaviour under experiements one, two and three. Unlike the rest of the transaction families, The \emph{Serial} execution is faster in \emph{Tree} scheduler in Voting transaction family. We have discovered that the entire voters list and the parties are list are accesses for any transaction in this family instead of the one particular voter and party addresses. We should note that the design of the transaction family (smart contract) plays a crucial role in the performance optimization of the blockchain. One can observe that the \emph{ADJ\_DAG} and \emph{LL\_DAG} still perform better as they use transactions to represent the dependency data structure unlike \emph{Tree} scheduler that uses addresses. 

The \emph{Smart Validator} framework efficiently verifies the DAG shared by the miner and eliminates the need to reconstruct the DAG at every block validator. The execution time of the smart validator and adjacency DAG scheduler will only vary in the dependency graph creation aspect. To highlight the optimization acheived through smart validator we have analyzed the dependency data structure (DS) creation and verification time for various schedulars in \Figref{combined_plots}. One can observe that the \emph{smart validator} takes the least execution time from the Figures~\ref{fig:Data_structure_1},~\ref{fig:Data_structure_2} and ~\ref{fig:Data_structure_3}. From observing \Figref{Data_structure_3} we can state that the \emph{smart validator} is stable against the variations in the dependency in the graph. The rest of the extensive experiments results including experiments on \emph{Intkey} and \emph{Insurance} transaction families can be found in appendix C \cite{piduguralla+:2023:arxiv}.
\section{Related research} 

In the past few years, blockchain technology has gained tremendous popularity and is used in a wide variety of fields. Although blockchains are capable of offering a variety of advantages, one of the most cited concerns is scalability. Consensus protocols and transaction throughput are the two major bottlenecks of blockchain performance. In contrast to proof of work (PoW), alternative consensus protocols like proof of stake (PoS) and proof of elapsed time (PoET) are introduced to minimize consensus time. However, transaction throughput continues to be a hindrance to scalability. Exercising parallel execution of transactions in a block is one of the solutions to optimize blockchains.

Dickerson et al. \cite{Dickerson+:ACSC:PODC:2017} introduced the concept of parallel execution of Ethereum \cite{ethereum} transactions using Software Transactional Memory (STM). The miners execute transactions in the block using STM, and the serializable concurrent state is discovered. This is then shared with the validators to achieve deterministic execution. Following this there have been multiple STM-based concurrent transaction execution frameworks for blockchains \cite{ParBlock:ICDCS:2019,parwat:springer:2021,blockSTM:2022,OptSmart:DPD:2022}. Besides the significant overhead associated with executing transactions through STMs, transactions sometimes fail due to dependencies and must be re-executed. Another drawback is that they cannot have operations that cannot be undone, which is a significant obstacle to smart contract design. During concurrent execution, STM-based approaches identify conflicts among transactions dynamically, i.e., during execution. This results in various transactions failing or rolling back to resolve the conflict. This has a significant impact on throughput and is not optimal for blocks with high interdependencies. In general, a dynamic approach is ideal, but it is not necessary for blockchains whose addresses are either included in the transactions or are easily inferred. For such systems, we propose a parallel execution framework for transactions in a block.

Sharding is another popular technique to address scaling issues in blockchains. In this, the nodes present in the network are categorized into small groups. Each group processes transactions parallelly with the other groups. Sharding is being explored earnestly as a solution to scalability issues \cite{shard:CCS:2016,Rapidchain:ACM:2018,Kokoris+:EEE:SP:2018,Zamani+:ACM:SIGSAC:2018,Dang+:SIGMOD:2019,SoCC:IEEE:2021,George+:ACM:PLDR:2021,DiPETrans:CPE:2022,Zheng+:IEEE:TII:2022}. The criteria for sharding are different in each approach. Few are specifically designed for monetary transactions in blockchains \cite{Kokoris+:EEE:SP:2018,Zamani+:ACM:SIGSAC:2018}. This leads to smart contract transactions being processed on a single shard leading to an inefficient distribution of computational work. The implementation of transactions that span across smart contracts becomes intricate with sharding. Protocols have to be designed specifically for inter-shard communication further increasing the complexity of the design \cite{Dang+:SIGMOD:2019}. The sharding technique limits the degree of parallelization to the number of shards irrespective of actual capacity. If the shards are dynamic, in the worst case the number of shards is equal to the number of transactions. This framework is unsuitable for transactions with high interdependencies.
\section{Conclusion and future work}
\label{apn:conclusion}
In this paper, we designed a framework for concurrent transaction execution in blockchains through static analysis of miners and validators of the network. This framework introduces parallel scheduler and smart validator modules into the blockchain node architecture. The parallel scheduler is responsible for identifying the dependencies among the transactions in the block and scheduling them for conflict-free execution. The determined dependencies are represented by a direct acyclic graph (DAG) and are shared along with the block to minimize the validation time of the blockchain. DAGs are validated using the smart validator, which determines if malicious miners have shared inaccurate graphs. Our framework has been extensively tested in Hyperledger Sawtooth 1.2.6. Next, we plan to improve fault tolerance and scalability for each blockchain node individually.
\bibliographystyle{splncs04}
\bibliography{citations}

\begin{thebibliography}{10}
\providecommand{\url}[1]{\texttt{#1}}
\providecommand{\urlprefix}{URL }
\providecommand{\doi}[1]{https://doi.org/#1}

\bibitem{ParBlock:ICDCS:2019}
Amiri, M.J., Agrawal, D., El~Abbadi, A.: Parblockchain: Leveraging transaction
  parallelism in permissioned blockchain systems. In: 2019 IEEE 39th
  International Conference on Distributed Computing Systems (ICDCS). pp.
  1337--1347 (2019). \doi{10.1109/ICDCS.2019.00134}

\bibitem{parwat:springer:2021}
Anjana, P.S., Attiya, H., Kumari, S., Peri, S., Somani, A.: Efficient
  concurrent execution of smart contracts in blockchains using object-based
  transactional memory. In: Georgiou, C., Majumdar, R. (eds.) Networked
  Systems. pp. 77--93. Springer International Publishing, Cham (2021)

\bibitem{parwat:pdp:2019}
Anjana, P.S., Kumari, S., Peri, S., Rathor, S., Somani, A.: An efficient
  framework for optimistic concurrent execution of smart contracts. In: 2019
  27th Euromicro International Conference on Parallel, Distributed and
  Network-Based Processing (PDP). pp. 83--92 (2019).
  \doi{10.1109/EMPDP.2019.8671637}

\bibitem{OptSmart:DPD:2022}
Anjana, P.S., Kumari, S., Peri, S., Rathor, S., Somani, A.: Optsmart: a space
  efficient optimistic concurrent execution of smart contracts. Distributed and
  Parallel Databases  (May 2022). \doi{10.1007/s10619-022-07412-y},
  \url{https://doi.org/10.1007/s10619-022-07412-y}

\bibitem{DiPETrans:CPE:2022}
Baheti, S., Anjana, P.S., Peri, S., Simmhan, Y.: Dipetrans: A framework for
  distributed parallel execution of transactions of blocks in blockchains.
  Concurrency and Computation: Practice and Experience  \textbf{34}(10),  e6804
  (2022). \doi{https://doi.org/10.1002/cpe.6804},
  \url{https://onlinelibrary.wiley.com/doi/abs/10.1002/cpe.6804}

\bibitem{Dang+:SIGMOD:2019}
Dang, H., Dinh, T.T.A., Loghin, D., Chang, E.C., Lin, Q., Ooi, B.C.: Towards
  scaling blockchain systems via sharding. In: Proceedings of the 2019
  International Conference on Management of Data. p. 123–140. SIGMOD '19,
  Association for Computing Machinery, New York, NY, USA (2019).
  \doi{10.1145/3299869.3319889}, \url{https://doi.org/10.1145/3299869.3319889}

\bibitem{Dickerson+:ACSC:PODC:2017}
Dickerson, T., Gazzillo, P., Herlihy, M., Koskinen, E.: Adding concurrency to
  smart contracts. p. 303–312. PODC '17, Association for Computing Machinery,
  New York, NY, USA (2017). \doi{10.1145/3087801.3087835},
  \url{https://doi.org/10.1145/3087801.3087835}

\bibitem{blockSTM:2022}
Gelashvili, R., Spiegelman, A., Xiang, Z., Danezis, G., Li, Z., Malkhi, D.,
  Xia, Y., Zhou, R.: Block-stm: Scaling blockchain execution by turning
  ordering curse to a performance blessing (2022).
  \doi{10.48550/ARXIV.2203.06871}, \url{https://arxiv.org/abs/2203.06871}

\bibitem{Kokoris+:EEE:SP:2018}
Kokoris-Kogias, E., Jovanovic, P., Gasser, L., Gailly, N., Syta, E., Ford, B.:
  Omniledger: A secure, scale-out, decentralized ledger via sharding. In: 2018
  IEEE Symposium on Security and Privacy (SP). pp. 583--598 (2018).
  \doi{10.1109/SP.2018.000-5}

\bibitem{KunBla+:1997:CJ}
Kunz, T., Black, J.P., Taylor, D.J., Basten, T.: {POET}: {T}arget-{S}ystem
  {I}ndependent {V}isualizations of {C}omplex {D}istributed-{A}pplications
  {E}xecutions. The Computer Journal  \textbf{40}(8) (1997)

\bibitem{shard:CCS:2016}
Luu, L., Narayanan, V., Zheng, C., Baweja, K., Gilbert, S., Saxena, P.: A
  secure sharding protocol for open blockchains. p. 17–30. CCS '16,
  Association for Computing Machinery, New York, NY, USA (2016).
  \doi{10.1145/2976749.2978389}, \url{https://doi.org/10.1145/2976749.2978389}

\bibitem{Nakamoto:Bitcoin:2009}
Nakamoto, S.: Bitcoin: A peer-to-peer electronic cash system (2009)

\bibitem{piduguralla+:2023:arxiv}
Piduguralla, M., Chakraborty, S., Anjana, P.S., Peri, S.: An efficient
  framework for execution of smart contracts in hyperledger sawtooth (2023).
  \doi{10.48550/ARXIV.2302.08452}, \url{https://arxiv.org/abs/2302.08452}

\bibitem{George+:ACM:PLDR:2021}
P\^{\i}rlea, G., Kumar, A., Sergey, I.: Cosplit (pldi 2021 artefact) (2021),
  \url{https://doi.org/10.5281/zenodo.4674301}

\bibitem{SoCC:IEEE:2021}
Valtchanov, A., Helbling, L., Mekiker, B., Wittie, M.P.: Parallel block
  execution in socc blockchains through optimistic concurrency control. In:
  2021 IEEE Globecom Workshops (GC Wkshps). pp.~1--6 (2021).
  \doi{10.1109/GCWkshps52748.2021.9682147}

\bibitem{Vasin:2014:WP}
Vasin, P.: Blackcoin’s proof-of-stake protocol v2. URL: https://blackcoin.
  co/blackcoin-pos-protocol-v2-whitepaper. pdf  \textbf{71} (2014)

\bibitem{Gerhard+:Book:2002}
Weikum, G., Vossen, G.: Chapter three - concurrency control: Notions of
  correctness for the page model. In: Weikum, G., Vossen, G. (eds.)
  Transactional Information Systems, pp. 61--123. The Morgan Kaufmann Series in
  Data Management Systems, Morgan Kaufmann, San Francisco (2002).
  \doi{https://doi.org/10.1016/B978-155860508-4/50004-1},
  \url{https://www.sciencedirect.com/science/article/pii/B9781558605084500041}

\bibitem{ethereum}
Wood, G.: Ethereum: A secure decentralised generalised transaction ledger

\bibitem{Rapidchain:ACM:2018}
Zamani, M., Movahedi, M., Raykova, M.: Rapidchain: Scaling blockchain via full
  sharding. In: Proceedings of the 2018 ACM SIGSAC Conference on Computer and
  Communications Security. p. 931–948. CCS '18, Association for Computing
  Machinery, New York, NY, USA (2018). \doi{10.1145/3243734.3243853},
  \url{https://doi.org/10.1145/3243734.3243853}

\bibitem{Zamani+:ACM:SIGSAC:2018}
Zamani, M., Movahedi, M., Raykova, M.: Rapidchain: Scaling blockchain via full
  sharding. In: Proceedings of the 2018 ACM SIGSAC Conference on Computer and
  Communications Security. p. 931–948. CCS '18, Association for Computing
  Machinery, New York, NY, USA (2018). \doi{10.1145/3243734.3243853},
  \url{https://doi.org/10.1145/3243734.3243853}

\bibitem{Zheng+:IEEE:TII:2022}
Zheng, P., Xu, Q., Luo, X., Zheng, Z., Zheng, W., Chen, X., Zhou, Z., Yan, Y.,
  Zhang, H.: Aeolus: Distributed execution of permissioned blockchain
  transactions via state sharding. IEEE Transactions on Industrial Informatics
  \textbf{18}(12),  9227--9238 (2022). \doi{10.1109/TII.2022.3164433}

\end{thebibliography}
\clearpage
\appendix
\section*{Appendix}
\label{apn:appendix}
\section{Hyperledger Sawtooth}
\subsection{Introduction}
Blockchain is a decentralized, distributed, immutable ledger system shared across a network introduced in 2008. The Hyperledger Foundation is “an open-source collaborative effort to advance cross-industry blockchain technologies. It is a global collaboration, hosted by The Linux Foundation, including leaders in finance, banking, IoT, supply chain, manufacturing, and technology” [1]. One of the most popular Hyperledger products is Hyperledger Sawtooth. A key goal of Sawtooth is to make Smart Contract execution safe, which is crucial for enterprise applications [2]. Hyperledger Sawtooth is gaining popularity for smart contract execution because of its highly modular approach. It facilitates the Sawtooth core design such that the transaction rules, permissions, and consensus algorithms can be customized according to the particular area of application.
Sawtooth has some distinctive features: 

\textbf{Separation:} The separation between the application level and the core system. With Sawtooth, users can build their models based on their business needs. Transaction rules, permissions, and consensus settings are handled in the Transaction Processing part. The Transaction Processors are the Server side of the application, while validators manage the core functionalities like verifying the transactions submitted using the particular consensus algorithm.

\textbf{Multi Language Support:} One of the striking features and the reason for the popularity of Sawtooth is the multiple language support, including Python, Javascript, Go, Rust, C++, etc., for the client and server sides. The language best suited for the desired operation can be chosen using one of the SDKs among the ones provided.
Parallel transaction Execution: Parallel transaction execution is one of the most sought-after features in blockchain to enhance speed and performance. Hyperledger Sawtooth facilitates that by splitting transactions into parallel flows.

\textbf{Consensus:} The consensus algorithm provides a way to ensure mutual agreement and trust among a group of participants. As Sawtooth separates consensus from transaction semantics, the user can implement multiple types of consensus on the same blockchain. In addition, it allows the user to change the consensus algorithm on a running blockchain. 
Other notable features of Sawtooth include permissioned and permissionless architecture, compatibility with Seth Ethereum Contracts.
\subsection{Architecture}
 In Sawtooth, the modular design separates the application and core levels. Sawtooth views smart contracts as a state machine or transaction processor [3]. By creating and applying transactions, the state can be modified. One or more transactions wrapped together form a batch. A batch is the atomic unit of change in the system. Multiple batches are combined to form a block (\Figref{block}). The principal modules observed in \Figref{Sawtooth} of the Sawtooth blockchain node are,
 \noindent
\textbf{Validator:} A validator is made up of the state management, transaction handling, and interconnection network components of a Sawtooth node.

 \noindent
\textbf{Global State:} Sawtooth uses an addressable Radix Merkle tree to store data for transaction families. Like a Merkle tree the successive node hashes are stored from leaf-to-root and like a Radix tree, the addresses uniquely identify the paths to leaf nodes.

 \noindent
\textbf{Journal:} The Block Completer, Block Validator, and Chain Controller modules make up the Journal of the node. All these modules work together to produce and validate blocks in the network. The blocks under processing are stored in BlockCache and BlockStore.

 \noindent
\textbf{Transaction Scheduler:} Transaction scheduler selects the transaction for execution. Sawtooth parallel scheduler utilizes a tree structure to keep track of dependencies using a Merkle tree. Each data location has read and write lists that contain the ID of the transactions that are accessing them.

 \noindent
\textbf{Transaction Executor:} The transaction executor invokes the scheduler to select a transaction for execution either for block production or block validation. The context is collected from the context manager and submitted along with the transaction to the transaction processor by the executor.

 \noindent
\textbf{Transaction Families:} Smart contracts are called transaction families in Hyperledger Sawtooth. The logic of the smart contracts is encoded into the transaction processor. Each smart contract has its own transaction processor and client. The transaction client checks the syntax of the transaction submitted and wraps them into a batch before submitting to the rest API.

 \noindent
\textbf{Rest API:} The clients interact with Sawtooth validator through Rest API module. Clients can submit transactions to the network or read the blocks added to the chain. Verifying the validity of the requests is the responsibility of the validator of the Sawtooth node while the Rest API acts as an intermediary.

The basic flow of execution can be, in short, described as:
\begin{enumerate}
    \item The client builds transactions and batches and signs them using the user’s private key.
    \item The explicit dependencies within a batch are indicated through dependencies filed in the transaction header.
    \item The transaction travels through a journal and is included in a block.
    \item The parallel scheduler calculates implicit dependencies.
    \item The batches are submitted to the validator via REST API.
    \item The validator combines the batches into blocks through the block publisher.
    \item Block publisher gives the batches to the scheduler.
    \item The scheduler submits the transactions to the executor.
    \item The executor gives the transaction and the related context to the transaction processor.
    \item The transaction processor validates the transaction using the business logic given
    \item Then the consensus protocol is implemented to see if the valid block should be added to the chain or not.
\end{enumerate}
\begin{figure}[]
\centering
\includegraphics[scale=0.35]{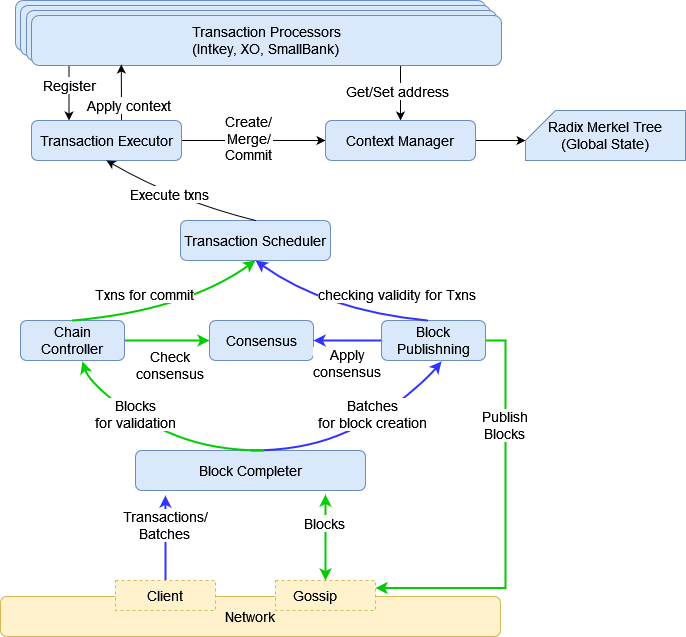}
\caption{Blockchain node architecture of Hyperledger Sawtooth}
\label{fig:Sawtooth}
\end{figure}
In \Figref{Sawtooth}, the blue path is the block production and the green path is the block validation process.
Transactions in a block in blockchains are executed by miners for block production and by validators for block validation. In most of the frameworks designed for concurrent execution, the model is the same for miners and validators. This is redundant as the analysis performed at the miners' end for concurrency is not utilized across the networks during validation. In the Optsmart framework \cite{OptSmart:DPD:2022} STMs are implemented for efficient concurrent execution and a DAG is used to represent conflicts between transactions. In this work, the use of the shared DAG created by the miner is presented. Additionally, it is discussed how it can be exploited by malicious miners to disrupt the network. A smart multi-threaded validator (SMV) \cite{parwat:pdp:2019} is proposed by Anjana et.al, which validates the DAG provided by the miner. During execution, the smart multi-threaded validator verifies the DAG shared dynamically by keeping track of transactions accessing various addresses through global and local counters. There exists a flaw in the SMV design as it is vulnerable to thread-serial execution (TSE). If fewer than the optimal number of threads are executing the transactions, the SMV will fail to identify missing edges. As the SMV identifies missing edges only when two or more transactions are accessing an address simultaneously. Assume the validators have low concurrent execution capability and did not execute the transactions in parallel. In that case, even if the transactions accessing the same address do not have edges, the SMV will not detect them. The SMV design does not account for extra edges introduced into the DAG by a malicious miner. Extra edges would result from forced serial execution of transactions. This would result in lower throughput. We present a framework with a parallel scheduler and a smart validator that overcomes the disadvantages discussed here.
\begin{figure}[]
\centering
\includegraphics[scale=0.25]{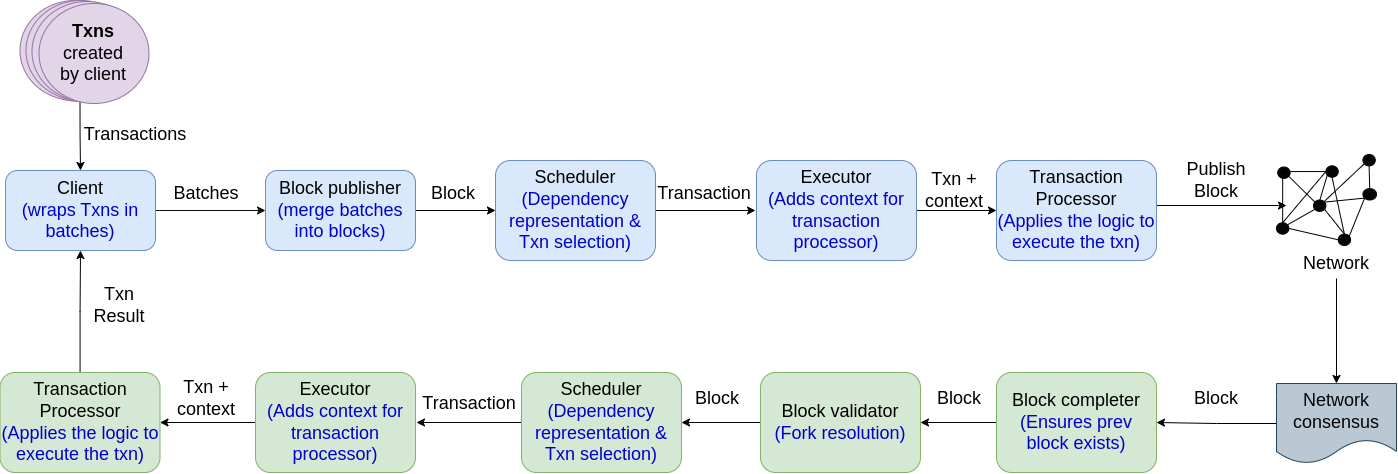}
\caption{Lifecycle of Sawtooth transaction}
\label{fig:Sawtooth}
\end{figure}

\section{Proofs}
	We aim to prove safety and liveliness of our proposed scheduler. We prove the safety by proving that the history of the execution is conflict serializable to the sequential execution of the transactions in the order they are present in the batch. The liveness of the scheduler is proved by proving that all transaction will eventually be committed.
	
	\begin{definition}[Dependency Graph]
	\label{def:DG}
		The dependency graph ($DG$) is constructed to represent the dependency information among the transactions in the block ($B$). The vertices of the graph $DG$ represent the transactions while the edges represent the dependency relation between the two transactions.
		
		$DG$ is defined as:
		
		$DG(B) = (V,E), $
		where, $\{t_1,t_2, ..., t_n\} \in V , $
		
		$(t_i, t_j) \in E \iff t_i \neq t_j \land (\exists p \in t_i) (\exists q \in t_j)$
		$(p,q) \in conf(B) \land i < j$.
		
		$p,q$ are operations in $t_i$ and $t_j$ respectively. $conf(B)$ are the read-write, write-write and write-read conflicting operations to the same address.
	\end{definition}
	\begin{lemma}
	\label{Lemma:acyclic}
		The dependency graph $DG$ is acyclic
	\end{lemma}
	\begin{proof}
	 Let us assume that $DG(B)$ is cyclic.
	 
	 This means that there a path from $t_i$ to $t_i$ through $t_j$ present in $DG(B)$.
	 
	 Let the path be $t_i,t_{k_1},t_{k_2},...t_j,...t_{k_n},t_i$
	 
	 $\Rightarrow t_i<t_{k_1}<t_{k_2}<...t_j<...t_{k_n}<t_i$ (Definition:~\ref{def:DG})
	 
	  This gives a contradiction, $t_i<t_j \land t_j<t_i \land i \neq j$ 
	  
	  Thus our initial assumption that $DG(B)$ not being acyclic is not true.
	  
	  We can conclude that $DG(B)$ is acyclic.
	 
	\end{proof}

	\begin{definition}[Possible histories ($H=DG(Gen)$)]
        \label{def:PS}
	The set $H$ is the set of all the possible histories using a given dependency graph $DG$.
	
	$ \forall h \in H, h$ is a resultant history of execution using $DG(B)$
        
        Given a $DG$ for a block B, the transactions with indegree zero ($t_i$) are scheduled for execution and are flagged (marked that the txn is in execution). Once the commit of the transaction is performed the outgoing edges from $t_i$ are removed from the graph.

        \end{definition}

	\begin{lemma}
	\label{lemma:topological_sort}
		For a given $h \in H$. $h$ is the topological sort of the graph $DG$
	\end{lemma}
	\begin{proof}
	    Topological sorting is a linear ordering of vertices of a graph such that if an edge is present between $p$ and $q$, the vertex $p$ comes before $q$ in the sort. The graph should be acyclic and directed.

        Based on Lemma:~\ref{Lemma:acyclic}, we know that $DG$ is acyclic. From the definition of Def:~\ref{def:PS}, we know that the possible histories contain all the transaction in the $DG$ and if an edge is present from $t_i$ to $t_j$, the indegree of $t_j$ would be greater than zero and $t_j$ will not be scheduled for execution unless $t_i$ is committed.
		
		This implies that all the operation of $t_i$ occur before any operations of $t_j$. Therefore $t_i$ is present present before $t_j$ in the history $h$.
        
	\end{proof}

	\begin{definition}[Conflict graph (G)]
	Let $s$ be a schedule. the conflict graph, also known as the serialization graph, $G(s) = (V,E) $ of $s$, is defined by
	
	$V = commit(s),$
	
	$(t,t') \in E \iff t \neq t' \land (\exists p \in t ) (\exists q \in t') (p,q) \in conf(s)$
	\end{definition}

	\begin{lemma}
		The dependency graph ($DG$) for a given block is isomorphic to the conflict graph of all the schedules of set S.
		\end{lemma}
		
		\begin{proof}
		    
		let $DG(B)=(V,E)$,
		
		$G(S)= (V',E')$
		
		$V=V'= \{t_1,t_2,...,t_n\}$
		
		if $ \exists (t_i,t_j) \in V' \Rightarrow (\exists p \in t_i ) (\exists q \in t_j) (p,q) \in conf(s)$
			
		$\Rightarrow (t_i,t_j) \in V$ (Lemma:~\ref{lemma:txn_order})
		
		This implies $E=E'$
		
		\end{proof}

	\begin{theorem}
	\label{theorem:CSR}
		Any history $h \in H$ is conflict equivalent to $h_0= {t_1,t_2,t_3....t_{n-1},t_n}$ i.e., . the sequential execution of ascending order of the transaction.
	\end{theorem}
	\begin{proof}
	  Two histories are conflict equivalent if they have same operations and same conflict relations \cite{Gerhard+:Book:2002}.
	  
	  This is true for all elements of $h$
	  
	  Therefore all the elements in $S$ are conflict equivalent to each other.$\forall s \in S, G(s)$ is isometric to $DG(B)$.
   
	\end{proof}

	\begin{theorem}
	\label{theorem:CSR}
		All the histories in set $S$ are conflict serializable (CSR).
	\end{theorem}
	\begin{proof}
	 Let $h$ be a history in $H$. $h \in CSR \iff CG(h)$ is acyclic \cite{Gerhard+:Book:2002}.
	 
	 We have show that $\forall h \in H, G(h)$ is isometric to $DG(B)$.
	 
	  We can conclude that $DG(B)$ is acyclic implying $h \in CSR$
	 
	\end{proof}

\section{Results}
\subsection{SimpleWallet Transaction Family}
\textbf{Experiment 1:} 
\begin{table}[]
\begin{tabular}{|c|cccc|cccc|cccc|}
\hline
\multirow{2}{*}{\textbf{\begin{tabular}[c]{@{}c@{}}No of \\ Blocks\end{tabular}}} & \multicolumn{4}{c|}{\textbf{CP1}}                                                                                                                                                                                       & \multicolumn{4}{c|}{\textbf{CP2}}                                                                                                                                                                                       & \multicolumn{4}{c|}{\textbf{CP3}}                                                                                                                                                                                       \\ \cline{2-13} 
                                                                                  & \multicolumn{1}{c|}{\textbf{\begin{tabular}[c]{@{}c@{}}ADJ\\ DAG\end{tabular}}} & \multicolumn{1}{c|}{\textbf{Tree}} & \multicolumn{1}{c|}{\textbf{\begin{tabular}[c]{@{}c@{}}LL\\ DAG\end{tabular}}} & \textbf{Serial} & \multicolumn{1}{c|}{\textbf{\begin{tabular}[c]{@{}c@{}}ADJ\\ DAG\end{tabular}}} & \multicolumn{1}{c|}{\textbf{Tree}} & \multicolumn{1}{c|}{\textbf{\begin{tabular}[c]{@{}c@{}}LL\\ DAG\end{tabular}}} & \textbf{Serial} & \multicolumn{1}{c|}{\textbf{\begin{tabular}[c]{@{}c@{}}ADJ\\ DAG\end{tabular}}} & \multicolumn{1}{c|}{\textbf{Tree}} & \multicolumn{1}{c|}{\textbf{\begin{tabular}[c]{@{}c@{}}LL\\ DAG\end{tabular}}} & \textbf{Serial} \\ \hline
\textbf{10}                                                                       & \multicolumn{1}{c|}{139}                                                        & \multicolumn{1}{c|}{158}           & \multicolumn{1}{c|}{142}                                                       & 241             & \multicolumn{1}{c|}{175}                                                        & \multicolumn{1}{c|}{181}           & \multicolumn{1}{c|}{179}                                                       & 212             & \multicolumn{1}{c|}{173}                                                        & \multicolumn{1}{c|}{202}           & \multicolumn{1}{c|}{177}                                                       & 275             \\ \hline
\textbf{20}                                                                       & \multicolumn{1}{c|}{292}                                                        & \multicolumn{1}{c|}{331}           & \multicolumn{1}{c|}{296}                                                       & 473             & \multicolumn{1}{c|}{343}                                                        & \multicolumn{1}{c|}{377}           & \multicolumn{1}{c|}{357}                                                       & 426             & \multicolumn{1}{c|}{349}                                                        & \multicolumn{1}{c|}{405}           & \multicolumn{1}{c|}{351}                                                       & 577             \\ \hline
\textbf{30}                                                                       & \multicolumn{1}{c|}{419}                                                        & \multicolumn{1}{c|}{474}           & \multicolumn{1}{c|}{431}                                                       & 887             & \multicolumn{1}{c|}{505}                                                        & \multicolumn{1}{c|}{564}           & \multicolumn{1}{c|}{534}                                                       & 633             & \multicolumn{1}{c|}{523}                                                        & \multicolumn{1}{c|}{599}           & \multicolumn{1}{c|}{523}                                                       & 908             \\ \hline
\textbf{40}                                                                       & \multicolumn{1}{c|}{580}                                                        & \multicolumn{1}{c|}{638}           & \multicolumn{1}{c|}{583}                                                       & 1103            & \multicolumn{1}{c|}{697}                                                        & \multicolumn{1}{c|}{720}           & \multicolumn{1}{c|}{716}                                                       & 879             & \multicolumn{1}{c|}{687}                                                        & \multicolumn{1}{c|}{811}           & \multicolumn{1}{c|}{700}                                                       & 1210            \\ \hline
\textbf{50}                                                                       & \multicolumn{1}{c|}{699}                                                        & \multicolumn{1}{c|}{834}           & \multicolumn{1}{c|}{725}                                                       & 1372            & \multicolumn{1}{c|}{852}                                                        & \multicolumn{1}{c|}{950}           & \multicolumn{1}{c|}{893}                                                       & 1076            & \multicolumn{1}{c|}{859}                                                        & \multicolumn{1}{c|}{986}           & \multicolumn{1}{c|}{881}                                                       & 1463            \\ \hline
\end{tabular}
\end{table}

\noindent
\textbf{Experiment 2:}
\begin{table}[]
\begin{tabular}{|c|cccc|cccc|cccc|}
\hline
\multirow{2}{*}{\textbf{\begin{tabular}[c]{@{}c@{}}Txns\\  per \\ Block\end{tabular}}} & \multicolumn{4}{c|}{\textbf{CP1}}                                                                                                                                                                                       & \multicolumn{4}{c|}{\textbf{CP2}}                                                                                                                                                                                       & \multicolumn{4}{c|}{\textbf{CP3}}                                                                                                                                                                                       \\ \cline{2-13} 
                                                                                       & \multicolumn{1}{c|}{\textbf{\begin{tabular}[c]{@{}c@{}}ADJ\\ DAG\end{tabular}}} & \multicolumn{1}{c|}{\textbf{Tree}} & \multicolumn{1}{c|}{\textbf{\begin{tabular}[c]{@{}c@{}}LL\\ DAG\end{tabular}}} & \textbf{Serial} & \multicolumn{1}{c|}{\textbf{\begin{tabular}[c]{@{}c@{}}ADJ\\ DAG\end{tabular}}} & \multicolumn{1}{c|}{\textbf{Tree}} & \multicolumn{1}{c|}{\textbf{\begin{tabular}[c]{@{}c@{}}LL\\ DAG\end{tabular}}} & \textbf{Serial} & \multicolumn{1}{c|}{\textbf{\begin{tabular}[c]{@{}c@{}}ADJ\\ DAG\end{tabular}}} & \multicolumn{1}{c|}{\textbf{Tree}} & \multicolumn{1}{c|}{\textbf{\begin{tabular}[c]{@{}c@{}}LL\\ DAG\end{tabular}}} & \textbf{Serial} \\ \hline
\textbf{200}                                                                           & \multicolumn{1}{c|}{54}                                                         & \multicolumn{1}{c|}{56}            & \multicolumn{1}{c|}{54}                                                        & 143             & \multicolumn{1}{c|}{65}                                                         & \multicolumn{1}{c|}{68}            & \multicolumn{1}{c|}{59}                                                        & 73              & \multicolumn{1}{c|}{67}                                                         & \multicolumn{1}{c|}{70}            & \multicolumn{1}{c|}{67}                                                        & 75              \\ \hline
\textbf{400}                                                                           & \multicolumn{1}{c|}{113}                                                        & \multicolumn{1}{c|}{121}           & \multicolumn{1}{c|}{96}                                                        & 181             & \multicolumn{1}{c|}{134}                                                        & \multicolumn{1}{c|}{143}           & \multicolumn{1}{c|}{124}                                                       & 153             & \multicolumn{1}{c|}{135}                                                        & \multicolumn{1}{c|}{142}           & \multicolumn{1}{c|}{134}                                                       & 156             \\ \hline
\textbf{600}                                                                           & \multicolumn{1}{c|}{163}                                                        & \multicolumn{1}{c|}{178}           & \multicolumn{1}{c|}{147}                                                       & 248             & \multicolumn{1}{c|}{205}                                                        & \multicolumn{1}{c|}{213}           & \multicolumn{1}{c|}{185}                                                       & 253             & \multicolumn{1}{c|}{200}                                                        & \multicolumn{1}{c|}{221}           & \multicolumn{1}{c|}{185}                                                       & 254             \\ \hline
\textbf{800}                                                                           & \multicolumn{1}{c|}{201}                                                        & \multicolumn{1}{c|}{253}           & \multicolumn{1}{c|}{195}                                                       & 282             & \multicolumn{1}{c|}{259}                                                        & \multicolumn{1}{c|}{290}           & \multicolumn{1}{c|}{244}                                                       & 318             & \multicolumn{1}{c|}{263}                                                        & \multicolumn{1}{c|}{303}           & \multicolumn{1}{c|}{250}                                                       & 326             \\ \hline
\textbf{1000}                                                                          & \multicolumn{1}{c|}{279}                                                        & \multicolumn{1}{c|}{332}           & \multicolumn{1}{c|}{252}                                                       & 375             & \multicolumn{1}{c|}{343}                                                        & \multicolumn{1}{c|}{366}           & \multicolumn{1}{c|}{313}                                                       & 413             & \multicolumn{1}{c|}{352}                                                        & \multicolumn{1}{c|}{405}           & \multicolumn{1}{c|}{329}                                                       & 426             \\ \hline
\textbf{1200}                                                                          & \multicolumn{1}{c|}{342}                                                        & \multicolumn{1}{c|}{403}           & \multicolumn{1}{c|}{323}                                                       & 493             & \multicolumn{1}{c|}{424}                                                        & \multicolumn{1}{c|}{467}           & \multicolumn{1}{c|}{430}                                                       & 505             & \multicolumn{1}{c|}{420}                                                        & \multicolumn{1}{c|}{503}           & \multicolumn{1}{c|}{413}                                                       & 538             \\ \hline
\end{tabular}
\end{table}

\noindent
\textbf{Experiment 3:} 
\begin{table}[]
\begin{tabular}{|c|cccc|cccc|cccc|}
\hline
\multirow{2}{*}{\textbf{\begin{tabular}[c]{@{}c@{}}Dependency\\ percentage\end{tabular}}} & \multicolumn{4}{c|}{\textbf{CP1}}                                                                                                                                                                                       & \multicolumn{4}{c|}{\textbf{CP2}}                                                                                                                                                                                       & \multicolumn{4}{c|}{\textbf{CP3}}                                                                                                                                                                                       \\ \cline{2-13} 
                                                                                          & \multicolumn{1}{c|}{\textbf{\begin{tabular}[c]{@{}c@{}}ADJ\\ DAG\end{tabular}}} & \multicolumn{1}{c|}{\textbf{Tree}} & \multicolumn{1}{c|}{\textbf{\begin{tabular}[c]{@{}c@{}}LL\\ DAG\end{tabular}}} & \textbf{Serial} & \multicolumn{1}{c|}{\textbf{\begin{tabular}[c]{@{}c@{}}ADJ\\ DAG\end{tabular}}} & \multicolumn{1}{c|}{\textbf{Tree}} & \multicolumn{1}{c|}{\textbf{\begin{tabular}[c]{@{}c@{}}LL\\ DAG\end{tabular}}} & \textbf{Serial} & \multicolumn{1}{c|}{\textbf{\begin{tabular}[c]{@{}c@{}}ADJ\\ DAG\end{tabular}}} & \multicolumn{1}{c|}{\textbf{Tree}} & \multicolumn{1}{c|}{\textbf{\begin{tabular}[c]{@{}c@{}}LL\\ DAG\end{tabular}}} & \textbf{Serial} \\ \hline
\textbf{0}                                                                                & \multicolumn{1}{c|}{372}                                                        & \multicolumn{1}{c|}{378}           & \multicolumn{1}{c|}{350}                                                       & 712             & \multicolumn{1}{c|}{364}                                                        & \multicolumn{1}{c|}{354}           & \multicolumn{1}{c|}{341}                                                       & 773             & \multicolumn{1}{c|}{363}                                                        & \multicolumn{1}{c|}{358}           & \multicolumn{1}{c|}{347}                                                       & 774             \\ \hline
\textbf{20}                                                                               & \multicolumn{1}{c|}{354}                                                        & \multicolumn{1}{c|}{393}           & \multicolumn{1}{c|}{335}                                                       & 786             & \multicolumn{1}{c|}{364}                                                        & \multicolumn{1}{c|}{438}           & \multicolumn{1}{c|}{356}                                                       & 767             & \multicolumn{1}{c|}{355}                                                        & \multicolumn{1}{c|}{482}           & \multicolumn{1}{c|}{322}                                                       & 726             \\ \hline
\textbf{40}                                                                               & \multicolumn{1}{c|}{365}                                                        & \multicolumn{1}{c|}{454}           & \multicolumn{1}{c|}{347}                                                       & 717             & \multicolumn{1}{c|}{361}                                                        & \multicolumn{1}{c|}{489}           & \multicolumn{1}{c|}{350}                                                       & 758             & \multicolumn{1}{c|}{359}                                                        & \multicolumn{1}{c|}{500}           & \multicolumn{1}{c|}{322}                                                       & 763             \\ \hline
\textbf{60}                                                                               & \multicolumn{1}{c|}{348}                                                        & \multicolumn{1}{c|}{501}           & \multicolumn{1}{c|}{351}                                                       & 742             & \multicolumn{1}{c|}{360}                                                        & \multicolumn{1}{c|}{513}           & \multicolumn{1}{c|}{344}                                                       & 760             & \multicolumn{1}{c|}{362}                                                        & \multicolumn{1}{c|}{503}           & \multicolumn{1}{c|}{327}                                                       & 748             \\ \hline
\textbf{80}                                                                               & \multicolumn{1}{c|}{354}                                                        & \multicolumn{1}{c|}{579}           & \multicolumn{1}{c|}{350}                                                       & 767             & \multicolumn{1}{c|}{350}                                                        & \multicolumn{1}{c|}{545}           & \multicolumn{1}{c|}{341}                                                       & 779             & \multicolumn{1}{c|}{352}                                                        & \multicolumn{1}{c|}{523}           & \multicolumn{1}{c|}{329}                                                       & 745             \\ \hline
\textbf{100}                                                                              & \multicolumn{1}{c|}{349}                                                        & \multicolumn{1}{c|}{624}           & \multicolumn{1}{c|}{336}                                                       & 708             & \multicolumn{1}{c|}{354}                                                        & \multicolumn{1}{c|}{571}           & \multicolumn{1}{c|}{382}                                                       & 726             & \multicolumn{1}{c|}{354}                                                        & \multicolumn{1}{c|}{664}           & \multicolumn{1}{c|}{395}                                                       & 755             \\ \hline
\end{tabular}
\end{table}

\newpage
\subsection{Intkey Transaction Family}
\textbf{Experiment 1:}
\begin{table}[]
\begin{tabular}{|c|cccc|cccc|cccc|}
\hline
\multirow{2}{*}{\textbf{\begin{tabular}[c]{@{}c@{}}No of \\ Blocks\end{tabular}}} & \multicolumn{4}{c|}{\textbf{CP1}}                                                                                                                                                                                       & \multicolumn{4}{c|}{\textbf{CP2}}                                                                                                                                                                                       & \multicolumn{4}{c|}{\textbf{CP3}}                                                                                                                                                                                       \\ \cline{2-13} 
                                                                                  & \multicolumn{1}{c|}{\textbf{\begin{tabular}[c]{@{}c@{}}ADJ\\ DAG\end{tabular}}} & \multicolumn{1}{c|}{\textbf{Tree}} & \multicolumn{1}{c|}{\textbf{\begin{tabular}[c]{@{}c@{}}LL\\ DAG\end{tabular}}} & \textbf{Serial} & \multicolumn{1}{c|}{\textbf{\begin{tabular}[c]{@{}c@{}}ADJ\\ DAG\end{tabular}}} & \multicolumn{1}{c|}{\textbf{Tree}} & \multicolumn{1}{c|}{\textbf{\begin{tabular}[c]{@{}c@{}}LL\\ DAG\end{tabular}}} & \textbf{Serial} & \multicolumn{1}{c|}{\textbf{\begin{tabular}[c]{@{}c@{}}ADJ\\ DAG\end{tabular}}} & \multicolumn{1}{c|}{\textbf{Tree}} & \multicolumn{1}{c|}{\textbf{\begin{tabular}[c]{@{}c@{}}LL\\ DAG\end{tabular}}} & \textbf{Serial} \\ \hline
\textbf{10}                                                                       & \multicolumn{1}{c|}{115}                                                        & \multicolumn{1}{c|}{126}           & \multicolumn{1}{c|}{118}                                                       & 200             & \multicolumn{1}{c|}{112}                                                        & \multicolumn{1}{c|}{135}           & \multicolumn{1}{c|}{117}                                                       & 138             & \multicolumn{1}{c|}{111}                                                        & \multicolumn{1}{c|}{133}           & \multicolumn{1}{c|}{115}                                                       & 132             \\ \hline
\textbf{20}                                                                       & \multicolumn{1}{c|}{231}                                                        & \multicolumn{1}{c|}{252}           & \multicolumn{1}{c|}{238}                                                       & 402             & \multicolumn{1}{c|}{224}                                                        & \multicolumn{1}{c|}{266}           & \multicolumn{1}{c|}{231}                                                       & 274             & \multicolumn{1}{c|}{223}                                                        & \multicolumn{1}{c|}{262}           & \multicolumn{1}{c|}{230}                                                       & 269             \\ \hline
\textbf{30}                                                                       & \multicolumn{1}{c|}{346}                                                        & \multicolumn{1}{c|}{374}           & \multicolumn{1}{c|}{358}                                                       & 597             & \multicolumn{1}{c|}{336}                                                        & \multicolumn{1}{c|}{394}           & \multicolumn{1}{c|}{346}                                                       & 412             & \multicolumn{1}{c|}{333}                                                        & \multicolumn{1}{c|}{391}           & \multicolumn{1}{c|}{349}                                                       & 398             \\ \hline
\textbf{40}                                                                       & \multicolumn{1}{c|}{459}                                                        & \multicolumn{1}{c|}{497}           & \multicolumn{1}{c|}{473}                                                       & 799             & \multicolumn{1}{c|}{448}                                                        & \multicolumn{1}{c|}{540}           & \multicolumn{1}{c|}{470}                                                       & 555             & \multicolumn{1}{c|}{446}                                                        & \multicolumn{1}{c|}{523}           & \multicolumn{1}{c|}{461}                                                       & 535             \\ \hline
\textbf{50}                                                                       & \multicolumn{1}{c|}{573}                                                        & \multicolumn{1}{c|}{631}           & \multicolumn{1}{c|}{597}                                                       & 1009            & \multicolumn{1}{c|}{562}                                                        & \multicolumn{1}{c|}{671}           & \multicolumn{1}{c|}{583}                                                       & 690             & \multicolumn{1}{c|}{551}                                                        & \multicolumn{1}{c|}{656}           & \multicolumn{1}{c|}{577}                                                       & 668             \\ \hline
\end{tabular}
\end{table}

\noindent
\textbf{Experiment 2:}
\begin{table}[]
\begin{tabular}{|c|cccc|cccc|cccc|}
\hline
\multirow{2}{*}{\textbf{\begin{tabular}[c]{@{}c@{}}Txns\\  per \\ Block\end{tabular}}} & \multicolumn{4}{c|}{\textbf{CP1}}                                                                                                                                                                                       & \multicolumn{4}{c|}{\textbf{CP2}}                                                                                                                                                                                       & \multicolumn{4}{c|}{\textbf{CP3}}                                                                                                                                                                                       \\ \cline{2-13} 
                                                                                       & \multicolumn{1}{c|}{\textbf{\begin{tabular}[c]{@{}c@{}}ADJ\\ DAG\end{tabular}}} & \multicolumn{1}{c|}{\textbf{Tree}} & \multicolumn{1}{c|}{\textbf{\begin{tabular}[c]{@{}c@{}}LL\\ DAG\end{tabular}}} & \textbf{Serial} & \multicolumn{1}{c|}{\textbf{\begin{tabular}[c]{@{}c@{}}ADJ\\ DAG\end{tabular}}} & \multicolumn{1}{c|}{\textbf{Tree}} & \multicolumn{1}{c|}{\textbf{\begin{tabular}[c]{@{}c@{}}LL\\ DAG\end{tabular}}} & \textbf{Serial} & \multicolumn{1}{c|}{\textbf{\begin{tabular}[c]{@{}c@{}}ADJ\\ DAG\end{tabular}}} & \multicolumn{1}{c|}{\textbf{Tree}} & \multicolumn{1}{c|}{\textbf{\begin{tabular}[c]{@{}c@{}}LL\\ DAG\end{tabular}}} & \textbf{Serial} \\ \hline
\textbf{200}                                                                           & \multicolumn{1}{c|}{39}                                                         & \multicolumn{1}{c|}{41}            & \multicolumn{1}{c|}{39}                                                        & 54              & \multicolumn{1}{c|}{39}                                                         & \multicolumn{1}{c|}{42}            & \multicolumn{1}{c|}{39}                                                        & 47              & \multicolumn{1}{c|}{39}                                                         & \multicolumn{1}{c|}{41}            & \multicolumn{1}{c|}{39}                                                        & 48              \\ \hline
\textbf{400}                                                                           & \multicolumn{1}{c|}{83}                                                         & \multicolumn{1}{c|}{84}            & \multicolumn{1}{c|}{87}                                                        & 121             & \multicolumn{1}{c|}{81}                                                         & \multicolumn{1}{c|}{94}            & \multicolumn{1}{c|}{84}                                                        & 106             & \multicolumn{1}{c|}{82}                                                         & \multicolumn{1}{c|}{91}            & \multicolumn{1}{c|}{89}                                                        & 96              \\ \hline
\textbf{600}                                                                           & \multicolumn{1}{c|}{127}                                                        & \multicolumn{1}{c|}{138}           & \multicolumn{1}{c|}{134}                                                       & 198             & \multicolumn{1}{c|}{122}                                                        & \multicolumn{1}{c|}{147}           & \multicolumn{1}{c|}{131}                                                       & 163             & \multicolumn{1}{c|}{126}                                                        & \multicolumn{1}{c|}{146}           & \multicolumn{1}{c|}{130}                                                       & 172             \\ \hline
\textbf{800}                                                                           & \multicolumn{1}{c|}{181}                                                        & \multicolumn{1}{c|}{191}           & \multicolumn{1}{c|}{188}                                                       & 281             & \multicolumn{1}{c|}{172}                                                        & \multicolumn{1}{c|}{205}           & \multicolumn{1}{c|}{181}                                                       & 214             & \multicolumn{1}{c|}{172}                                                        & \multicolumn{1}{c|}{203}           & \multicolumn{1}{c|}{179}                                                       & 201             \\ \hline
\textbf{1000}                                                                          & \multicolumn{1}{c|}{217}                                                        & \multicolumn{1}{c|}{248}           & \multicolumn{1}{c|}{227}                                                       & 395             & \multicolumn{1}{c|}{221}                                                        & \multicolumn{1}{c|}{261}           & \multicolumn{1}{c|}{232}                                                       & 274             & \multicolumn{1}{c|}{224}                                                        & \multicolumn{1}{c|}{255}           & \multicolumn{1}{c|}{228}                                                       & 259             \\ \hline
\textbf{1200}                                                                          & \multicolumn{1}{c|}{268}                                                        & \multicolumn{1}{c|}{302}           & \multicolumn{1}{c|}{281}                                                       & 573             & \multicolumn{1}{c|}{280}                                                        & \multicolumn{1}{c|}{327}           & \multicolumn{1}{c|}{285}                                                       & 345             & \multicolumn{1}{c|}{229}                                                        & \multicolumn{1}{c|}{264}           & \multicolumn{1}{c|}{237}                                                       & 263             \\ \hline
\end{tabular}
\end{table}

\noindent
\textbf{Experiment 3:} 
\begin{table}[]
\begin{tabular}{|c|cccc|cccc|cccc|}
\hline
\multirow{2}{*}{\textbf{\begin{tabular}[c]{@{}c@{}}Dependency\\ percentage\end{tabular}}} & \multicolumn{4}{c|}{\textbf{CP1}}                                                                                                                                                                                       & \multicolumn{4}{c|}{\textbf{CP2}}                                                                                                                                                                                       & \multicolumn{4}{c|}{\textbf{CP3}}                                                                                                                                                                                       \\ \cline{2-13} 
                                                                                          & \multicolumn{1}{c|}{\textbf{\begin{tabular}[c]{@{}c@{}}ADJ\\ DAG\end{tabular}}} & \multicolumn{1}{c|}{\textbf{Tree}} & \multicolumn{1}{c|}{\textbf{\begin{tabular}[c]{@{}c@{}}LL\\ DAG\end{tabular}}} & \textbf{Serial} & \multicolumn{1}{c|}{\textbf{\begin{tabular}[c]{@{}c@{}}ADJ\\ DAG\end{tabular}}} & \multicolumn{1}{c|}{\textbf{Tree}} & \multicolumn{1}{c|}{\textbf{\begin{tabular}[c]{@{}c@{}}LL\\ DAG\end{tabular}}} & \textbf{Serial} & \multicolumn{1}{c|}{\textbf{\begin{tabular}[c]{@{}c@{}}ADJ\\ DAG\end{tabular}}} & \multicolumn{1}{c|}{\textbf{Tree}} & \multicolumn{1}{c|}{\textbf{\begin{tabular}[c]{@{}c@{}}LL\\ DAG\end{tabular}}} & \textbf{Serial} \\ \hline
\textbf{0}                                                                                & \multicolumn{1}{c|}{224}                                                        & \multicolumn{1}{c|}{226}           & \multicolumn{1}{c|}{228}                                                       & 517             & \multicolumn{1}{c|}{226}                                                        & \multicolumn{1}{c|}{228}           & \multicolumn{1}{c|}{227}                                                       & 504             & \multicolumn{1}{c|}{227}                                                        & \multicolumn{1}{c|}{239}           & \multicolumn{1}{c|}{227}                                                       & 246             \\ \hline
\textbf{20}                                                                               & \multicolumn{1}{c|}{227}                                                        & \multicolumn{1}{c|}{242}           & \multicolumn{1}{c|}{231}                                                       & 436             & \multicolumn{1}{c|}{228}                                                        & \multicolumn{1}{c|}{249}           & \multicolumn{1}{c|}{229}                                                       & 352             & \multicolumn{1}{c|}{229}                                                        & \multicolumn{1}{c|}{242}           & \multicolumn{1}{c|}{231}                                                       & 322             \\ \hline
\textbf{40}                                                                               & \multicolumn{1}{c|}{229}                                                        & \multicolumn{1}{c|}{244}           & \multicolumn{1}{c|}{232}                                                       & 424             & \multicolumn{1}{c|}{229}                                                        & \multicolumn{1}{c|}{253}           & \multicolumn{1}{c|}{230}                                                       & 318             & \multicolumn{1}{c|}{230}                                                        & \multicolumn{1}{c|}{245}           & \multicolumn{1}{c|}{233}                                                       & 346             \\ \hline
\textbf{60}                                                                               & \multicolumn{1}{c|}{230}                                                        & \multicolumn{1}{c|}{249}           & \multicolumn{1}{c|}{235}                                                       & 383             & \multicolumn{1}{c|}{231}                                                        & \multicolumn{1}{c|}{264}           & \multicolumn{1}{c|}{233}                                                       & 274             & \multicolumn{1}{c|}{235}                                                        & \multicolumn{1}{c|}{247}           & \multicolumn{1}{c|}{236}                                                       & 382             \\ \hline
\textbf{80}                                                                               & \multicolumn{1}{c|}{233}                                                        & \multicolumn{1}{c|}{252}           & \multicolumn{1}{c|}{237}                                                       & 333             & \multicolumn{1}{c|}{233}                                                        & \multicolumn{1}{c|}{272}           & \multicolumn{1}{c|}{235}                                                       & 264             & \multicolumn{1}{c|}{238}                                                        & \multicolumn{1}{c|}{251}           & \multicolumn{1}{c|}{239}                                                       & 403             \\ \hline
\textbf{100}                                                                              & \multicolumn{1}{c|}{236}                                                        & \multicolumn{1}{c|}{254}           & \multicolumn{1}{c|}{250}                                                       & 315             & \multicolumn{1}{c|}{234}                                                        & \multicolumn{1}{c|}{279}           & \multicolumn{1}{c|}{245}                                                       & 249             & \multicolumn{1}{c|}{240}                                                        & \multicolumn{1}{c|}{254}           & \multicolumn{1}{c|}{244}                                                       & 449             \\ \hline
\end{tabular}
\end{table}

\newpage
\subsection{Voting Transaction Family}
\textbf{Experiment 1:}
\begin{table}[]
\begin{tabular}{|c|cccc|cccc|cccc|}
\hline
\multirow{2}{*}{\textbf{\begin{tabular}[c]{@{}c@{}}No of \\ Blocks\end{tabular}}} & \multicolumn{4}{c|}{\textbf{CP1}}                                                                                                                                                                                       & \multicolumn{4}{c|}{\textbf{CP2}}                                                                                                                                                                                       & \multicolumn{4}{c|}{\textbf{CP3}}                                                                                                                                                                                       \\ \cline{2-13} 
                                                                                  & \multicolumn{1}{c|}{\textbf{\begin{tabular}[c]{@{}c@{}}ADJ\\ DAG\end{tabular}}} & \multicolumn{1}{c|}{\textbf{Tree}} & \multicolumn{1}{c|}{\textbf{\begin{tabular}[c]{@{}c@{}}LL\\ DAG\end{tabular}}} & \textbf{Serial} & \multicolumn{1}{c|}{\textbf{\begin{tabular}[c]{@{}c@{}}ADJ\\ DAG\end{tabular}}} & \multicolumn{1}{c|}{\textbf{Tree}} & \multicolumn{1}{c|}{\textbf{\begin{tabular}[c]{@{}c@{}}LL\\ DAG\end{tabular}}} & \textbf{Serial} & \multicolumn{1}{c|}{\textbf{\begin{tabular}[c]{@{}c@{}}ADJ\\ DAG\end{tabular}}} & \multicolumn{1}{c|}{\textbf{Tree}} & \multicolumn{1}{c|}{\textbf{\begin{tabular}[c]{@{}c@{}}LL\\ DAG\end{tabular}}} & \textbf{Serial} \\ \hline
\textbf{10}                                                                       & \multicolumn{1}{c|}{81}                                                         & \multicolumn{1}{c|}{160}           & \multicolumn{1}{c|}{116}                                                       & 128             & \multicolumn{1}{c|}{81}                                                         & \multicolumn{1}{c|}{187}           & \multicolumn{1}{c|}{128}                                                       & 151             & \multicolumn{1}{c|}{80}                                                         & \multicolumn{1}{c|}{183}           & \multicolumn{1}{c|}{122}                                                       & 154             \\ \hline
\textbf{20}                                                                       & \multicolumn{1}{c|}{162}                                                        & \multicolumn{1}{c|}{318}           & \multicolumn{1}{c|}{247}                                                       & 254             & \multicolumn{1}{c|}{162}                                                        & \multicolumn{1}{c|}{368}           & \multicolumn{1}{c|}{255}                                                       & 291             & \multicolumn{1}{c|}{161}                                                        & \multicolumn{1}{c|}{352}           & \multicolumn{1}{c|}{239}                                                       & 312             \\ \hline
\textbf{30}                                                                       & \multicolumn{1}{c|}{241}                                                        & \multicolumn{1}{c|}{469}           & \multicolumn{1}{c|}{366}                                                       & 382             & \multicolumn{1}{c|}{239}                                                        & \multicolumn{1}{c|}{556}           & \multicolumn{1}{c|}{386}                                                       & 447             & \multicolumn{1}{c|}{243}                                                        & \multicolumn{1}{c|}{535}           & \multicolumn{1}{c|}{364}                                                       & 458             \\ \hline
\textbf{40}                                                                       & \multicolumn{1}{c|}{323}                                                        & \multicolumn{1}{c|}{646}           & \multicolumn{1}{c|}{485}                                                       & 511             & \multicolumn{1}{c|}{317}                                                        & \multicolumn{1}{c|}{749}           & \multicolumn{1}{c|}{509}                                                       & 596             & \multicolumn{1}{c|}{323}                                                        & \multicolumn{1}{c|}{746}           & \multicolumn{1}{c|}{478}                                                       & 616             \\ \hline
\textbf{50}                                                                       & \multicolumn{1}{c|}{404}                                                        & \multicolumn{1}{c|}{797}           & \multicolumn{1}{c|}{607}                                                       & 623             & \multicolumn{1}{c|}{396}                                                        & \multicolumn{1}{c|}{933}           & \multicolumn{1}{c|}{633}                                                       & 728             & \multicolumn{1}{c|}{399}                                                        & \multicolumn{1}{c|}{937}           & \multicolumn{1}{c|}{604}                                                       & 776             \\ \hline
\end{tabular}
\end{table}

\noindent
\textbf{Experiment 2:}
\begin{table}[]
\begin{tabular}{|c|cccc|cccc|cccc|}
\hline
\multirow{2}{*}{\textbf{\begin{tabular}[c]{@{}c@{}}Txns\\  per \\ Block\end{tabular}}} & \multicolumn{4}{c|}{\textbf{CP1}}                                                                                                                                                                                       & \multicolumn{4}{c|}{\textbf{CP2}}                                                                                                                                                                                       & \multicolumn{4}{c|}{\textbf{CP3}}                                                                                                                                                                                       \\ \cline{2-13} 
                                                                                       & \multicolumn{1}{c|}{\textbf{\begin{tabular}[c]{@{}c@{}}ADJ\\ DAG\end{tabular}}} & \multicolumn{1}{c|}{\textbf{Tree}} & \multicolumn{1}{c|}{\textbf{\begin{tabular}[c]{@{}c@{}}LL\\ DAG\end{tabular}}} & \textbf{Serial} & \multicolumn{1}{c|}{\textbf{\begin{tabular}[c]{@{}c@{}}ADJ\\ DAG\end{tabular}}} & \multicolumn{1}{c|}{\textbf{Tree}} & \multicolumn{1}{c|}{\textbf{\begin{tabular}[c]{@{}c@{}}LL\\ DAG\end{tabular}}} & \textbf{Serial} & \multicolumn{1}{c|}{\textbf{\begin{tabular}[c]{@{}c@{}}ADJ\\ DAG\end{tabular}}} & \multicolumn{1}{c|}{\textbf{Tree}} & \multicolumn{1}{c|}{\textbf{\begin{tabular}[c]{@{}c@{}}LL\\ DAG\end{tabular}}} & \textbf{Serial} \\ \hline
\textbf{200}                                                                           & \multicolumn{1}{c|}{27}                                                         & \multicolumn{1}{c|}{46}            & \multicolumn{1}{c|}{19}                                                        & 41              & \multicolumn{1}{c|}{27}                                                         & \multicolumn{1}{c|}{44}            & \multicolumn{1}{c|}{20}                                                        & 40              & \multicolumn{1}{c|}{28}                                                         & \multicolumn{1}{c|}{45}            & \multicolumn{1}{c|}{20}                                                        & 40              \\ \hline
\textbf{400}                                                                           & \multicolumn{1}{c|}{59}                                                         & \multicolumn{1}{c|}{102}           & \multicolumn{1}{c|}{42}                                                        & 90              & \multicolumn{1}{c|}{56}                                                         & \multicolumn{1}{c|}{102}           & \multicolumn{1}{c|}{45}                                                        & 87              & \multicolumn{1}{c|}{59}                                                         & \multicolumn{1}{c|}{103}           & \multicolumn{1}{c|}{44}                                                        & 88              \\ \hline
\textbf{600}                                                                           & \multicolumn{1}{c|}{90}                                                         & \multicolumn{1}{c|}{166}           & \multicolumn{1}{c|}{74}                                                        & 139             & \multicolumn{1}{c|}{91}                                                         & \multicolumn{1}{c|}{175}           & \multicolumn{1}{c|}{85}                                                        & 144             & \multicolumn{1}{c|}{91}                                                         & \multicolumn{1}{c|}{178}           & \multicolumn{1}{c|}{80}                                                        & 148             \\ \hline
\textbf{800}                                                                           & \multicolumn{1}{c|}{126}                                                        & \multicolumn{1}{c|}{236}           & \multicolumn{1}{c|}{131}                                                       & 189             & \multicolumn{1}{c|}{124}                                                        & \multicolumn{1}{c|}{265}           & \multicolumn{1}{c|}{149}                                                       & 212             & \multicolumn{1}{c|}{121}                                                        & \multicolumn{1}{c|}{265}           & \multicolumn{1}{c|}{140}                                                       & 224             \\ \hline
\textbf{1000}                                                                          & \multicolumn{1}{c|}{160}                                                        & \multicolumn{1}{c|}{317}           & \multicolumn{1}{c|}{212}                                                       & 250             & \multicolumn{1}{c|}{160}                                                        & \multicolumn{1}{c|}{370}           & \multicolumn{1}{c|}{246}                                                       & 299             & \multicolumn{1}{c|}{157}                                                        & \multicolumn{1}{c|}{345}           & \multicolumn{1}{c|}{230}                                                       & 307             \\ \hline
\textbf{1200}                                                                          & \multicolumn{1}{c|}{199}                                                        & \multicolumn{1}{c|}{414}           & \multicolumn{1}{c|}{351}                                                       & 400             & \multicolumn{1}{c|}{200}                                                        & \multicolumn{1}{c|}{502}           & \multicolumn{1}{c|}{390}                                                       & 409             & \multicolumn{1}{c|}{196}                                                        & \multicolumn{1}{c|}{488}           & \multicolumn{1}{c|}{359}                                                       & 399             \\ \hline
\end{tabular}
\end{table}

\noindent
\textbf{Experiment 3:} 
\begin{table}[]
\begin{tabular}{|c|cccc|cccc|cccc|}
\hline
\multirow{2}{*}{\textbf{\begin{tabular}[c]{@{}c@{}}Dependency\\ percentage\end{tabular}}} & \multicolumn{4}{c|}{\textbf{CP1}}                                                                                                                                                                                       & \multicolumn{4}{c|}{\textbf{CP2}}                                                                                                                                                                                       & \multicolumn{4}{c|}{\textbf{CP3}}                                                                                                                                                                                       \\ \cline{2-13} 
                                                                                          & \multicolumn{1}{c|}{\textbf{\begin{tabular}[c]{@{}c@{}}ADJ\\ DAG\end{tabular}}} & \multicolumn{1}{c|}{\textbf{Tree}} & \multicolumn{1}{c|}{\textbf{\begin{tabular}[c]{@{}c@{}}LL\\ DAG\end{tabular}}} & \textbf{Serial} & \multicolumn{1}{c|}{\textbf{\begin{tabular}[c]{@{}c@{}}ADJ\\ DAG\end{tabular}}} & \multicolumn{1}{c|}{\textbf{Tree}} & \multicolumn{1}{c|}{\textbf{\begin{tabular}[c]{@{}c@{}}LL\\ DAG\end{tabular}}} & \textbf{Serial} & \multicolumn{1}{c|}{\textbf{\begin{tabular}[c]{@{}c@{}}ADJ\\ DAG\end{tabular}}} & \multicolumn{1}{c|}{\textbf{Tree}} & \multicolumn{1}{c|}{\textbf{\begin{tabular}[c]{@{}c@{}}LL\\ DAG\end{tabular}}} & \textbf{Serial} \\ \hline
\textbf{0}                                                                                & \multicolumn{1}{c|}{157}                                                        & \multicolumn{1}{c|}{460}           & \multicolumn{1}{c|}{230}                                                       & 408             & \multicolumn{1}{c|}{155}                                                        & \multicolumn{1}{c|}{457}           & \multicolumn{1}{c|}{234}                                                       & 408             & \multicolumn{1}{c|}{156}                                                        & \multicolumn{1}{c|}{466}           & \multicolumn{1}{c|}{230}                                                       & 407             \\ \hline
\textbf{20}                                                                               & \multicolumn{1}{c|}{159}                                                        & \multicolumn{1}{c|}{456}           & \multicolumn{1}{c|}{233}                                                       & 398             & \multicolumn{1}{c|}{157}                                                        & \multicolumn{1}{c|}{427}           & \multicolumn{1}{c|}{240}                                                       & 365             & \multicolumn{1}{c|}{158}                                                        & \multicolumn{1}{c|}{441}           & \multicolumn{1}{c|}{230}                                                       & 392             \\ \hline
\textbf{40}                                                                               & \multicolumn{1}{c|}{160}                                                        & \multicolumn{1}{c|}{434}           & \multicolumn{1}{c|}{233}                                                       & 372             & \multicolumn{1}{c|}{158}                                                        & \multicolumn{1}{c|}{382}           & \multicolumn{1}{c|}{246}                                                       & 307             & \multicolumn{1}{c|}{158}                                                        & \multicolumn{1}{c|}{401}           & \multicolumn{1}{c|}{231}                                                       & 343             \\ \hline
\textbf{60}                                                                               & \multicolumn{1}{c|}{161}                                                        & \multicolumn{1}{c|}{396}           & \multicolumn{1}{c|}{246}                                                       & 331             & \multicolumn{1}{c|}{159}                                                        & \multicolumn{1}{c|}{339}           & \multicolumn{1}{c|}{255}                                                       & 271             & \multicolumn{1}{c|}{161}                                                        & \multicolumn{1}{c|}{363}           & \multicolumn{1}{c|}{231}                                                       & 291             \\ \hline
\textbf{80}                                                                               & \multicolumn{1}{c|}{162}                                                        & \multicolumn{1}{c|}{338}           & \multicolumn{1}{c|}{256}                                                       & 265             & \multicolumn{1}{c|}{160}                                                        & \multicolumn{1}{c|}{294}           & \multicolumn{1}{c|}{266}                                                       & 225             & \multicolumn{1}{c|}{163}                                                        & \multicolumn{1}{c|}{300}           & \multicolumn{1}{c|}{234}                                                       & 226             \\ \hline
\textbf{100}                                                                              & \multicolumn{1}{c|}{162}                                                        & \multicolumn{1}{c|}{273}           & \multicolumn{1}{c|}{250}                                                       & 175             & \multicolumn{1}{c|}{160}                                                        & \multicolumn{1}{c|}{272}           & \multicolumn{1}{c|}{250}                                                       & 169             & \multicolumn{1}{c|}{164}                                                        & \multicolumn{1}{c|}{274}           & \multicolumn{1}{c|}{249}                                                       & 174             \\ \hline
\end{tabular}
\end{table}

\newpage
\subsection{Insurance Transaction Family}
\textbf{Experiment 1:}
\begin{table}[]
\begin{tabular}{|c|cccc|cccc|cccc|}
\hline
\multirow{2}{*}{\textbf{\begin{tabular}[c]{@{}c@{}}No of \\ Blocks\end{tabular}}} & \multicolumn{4}{c|}{\textbf{CP1}}                                                                                                                                                                                       & \multicolumn{4}{c|}{\textbf{CP2}}                                                                                                                                                                                       & \multicolumn{4}{c|}{\textbf{CP3}}                                                                                                                                                                                       \\ \cline{2-13} 
                                                                                  & \multicolumn{1}{c|}{\textbf{\begin{tabular}[c]{@{}c@{}}ADJ\\ DAG\end{tabular}}} & \multicolumn{1}{c|}{\textbf{Tree}} & \multicolumn{1}{c|}{\textbf{\begin{tabular}[c]{@{}c@{}}LL\\ DAG\end{tabular}}} & \textbf{Serial} & \multicolumn{1}{c|}{\textbf{\begin{tabular}[c]{@{}c@{}}ADJ\\ DAG\end{tabular}}} & \multicolumn{1}{c|}{\textbf{Tree}} & \multicolumn{1}{c|}{\textbf{\begin{tabular}[c]{@{}c@{}}LL\\ DAG\end{tabular}}} & \textbf{Serial} & \multicolumn{1}{c|}{\textbf{\begin{tabular}[c]{@{}c@{}}ADJ\\ DAG\end{tabular}}} & \multicolumn{1}{c|}{\textbf{Tree}} & \multicolumn{1}{c|}{\textbf{\begin{tabular}[c]{@{}c@{}}LL\\ DAG\end{tabular}}} & \textbf{Serial} \\ \hline
\textbf{10}                                                                       & \multicolumn{1}{c|}{46}                                                         & \multicolumn{1}{c|}{51}            & \multicolumn{1}{c|}{44}                                                        & 154             & \multicolumn{1}{c|}{46}                                                         & \multicolumn{1}{c|}{56}            & \multicolumn{1}{c|}{46}                                                        & 92              & \multicolumn{1}{c|}{45}                                                         & \multicolumn{1}{c|}{55}            & \multicolumn{1}{c|}{44}                                                        & 88              \\ \hline
\textbf{20}                                                                       & \multicolumn{1}{c|}{96}                                                         & \multicolumn{1}{c|}{97}            & \multicolumn{1}{c|}{82}                                                        & 312             & \multicolumn{1}{c|}{93}                                                         & \multicolumn{1}{c|}{111}           & \multicolumn{1}{c|}{95}                                                        & 188             & \multicolumn{1}{c|}{91}                                                         & \multicolumn{1}{c|}{111}           & \multicolumn{1}{c|}{88}                                                        & 175             \\ \hline
\textbf{30}                                                                       & \multicolumn{1}{c|}{143}                                                        & \multicolumn{1}{c|}{155}           & \multicolumn{1}{c|}{137}                                                       & 466             & \multicolumn{1}{c|}{133}                                                        & \multicolumn{1}{c|}{169}           & \multicolumn{1}{c|}{141}                                                       & 275             & \multicolumn{1}{c|}{140}                                                        & \multicolumn{1}{c|}{169}           & \multicolumn{1}{c|}{128}                                                       & 258             \\ \hline
\textbf{40}                                                                       & \multicolumn{1}{c|}{183}                                                        & \multicolumn{1}{c|}{191}           & \multicolumn{1}{c|}{175}                                                       & 635             & \multicolumn{1}{c|}{188}                                                        & \multicolumn{1}{c|}{227}           & \multicolumn{1}{c|}{186}                                                       & 364             & \multicolumn{1}{c|}{177}                                                        & \multicolumn{1}{c|}{227}           & \multicolumn{1}{c|}{174}                                                       & 340             \\ \hline
\textbf{50}                                                                       & \multicolumn{1}{c|}{226}                                                        & \multicolumn{1}{c|}{238}           & \multicolumn{1}{c|}{214}                                                       & 769             & \multicolumn{1}{c|}{234}                                                        & \multicolumn{1}{c|}{278}           & \multicolumn{1}{c|}{236}                                                       & 466             & \multicolumn{1}{c|}{233}                                                        & \multicolumn{1}{c|}{278}           & \multicolumn{1}{c|}{217}                                                       & 438             \\ \hline
\end{tabular}
\end{table}

\noindent
\textbf{Experiment 2:}
\begin{table}[]
\begin{tabular}{|c|cccc|cccc|cccc|}
\hline
\multirow{2}{*}{\textbf{\begin{tabular}[c]{@{}c@{}}Txns\\  per \\ Block\end{tabular}}} & \multicolumn{4}{c|}{\textbf{CP1}}                                                                                                                                                                                       & \multicolumn{4}{c|}{\textbf{CP2}}                                                                                                                                                                                       & \multicolumn{4}{c|}{\textbf{CP3}}                                                                                                                                                                                       \\ \cline{2-13} 
                                                                                       & \multicolumn{1}{c|}{\textbf{\begin{tabular}[c]{@{}c@{}}ADJ\\ DAG\end{tabular}}} & \multicolumn{1}{c|}{\textbf{Tree}} & \multicolumn{1}{c|}{\textbf{\begin{tabular}[c]{@{}c@{}}LL\\ DAG\end{tabular}}} & \textbf{Serial} & \multicolumn{1}{c|}{\textbf{\begin{tabular}[c]{@{}c@{}}ADJ\\ DAG\end{tabular}}} & \multicolumn{1}{c|}{\textbf{Tree}} & \multicolumn{1}{c|}{\textbf{\begin{tabular}[c]{@{}c@{}}LL\\ DAG\end{tabular}}} & \textbf{Serial} & \multicolumn{1}{c|}{\textbf{\begin{tabular}[c]{@{}c@{}}ADJ\\ DAG\end{tabular}}} & \multicolumn{1}{c|}{\textbf{Tree}} & \multicolumn{1}{c|}{\textbf{\begin{tabular}[c]{@{}c@{}}LL\\ DAG\end{tabular}}} & \textbf{Serial} \\ \hline
\textbf{200}                                                                           & \multicolumn{1}{c|}{15}                                                         & \multicolumn{1}{c|}{16}            & \multicolumn{1}{c|}{14}                                                        & 40              & \multicolumn{1}{c|}{14}                                                         & \multicolumn{1}{c|}{20}            & \multicolumn{1}{c|}{13}                                                        & 36              & \multicolumn{1}{c|}{14}                                                         & \multicolumn{1}{c|}{20}            & \multicolumn{1}{c|}{13}                                                        & 35              \\ \hline
\textbf{400}                                                                           & \multicolumn{1}{c|}{30}                                                         & \multicolumn{1}{c|}{32}            & \multicolumn{1}{c|}{32}                                                        & 90              & \multicolumn{1}{c|}{30}                                                         & \multicolumn{1}{c|}{42}            & \multicolumn{1}{c|}{30}                                                        & 70              & \multicolumn{1}{c|}{31}                                                         & \multicolumn{1}{c|}{42}            & \multicolumn{1}{c|}{30}                                                        & 70              \\ \hline
\textbf{600}                                                                           & \multicolumn{1}{c|}{50}                                                         & \multicolumn{1}{c|}{49}            & \multicolumn{1}{c|}{46}                                                        & 148             & \multicolumn{1}{c|}{47}                                                         & \multicolumn{1}{c|}{68}            & \multicolumn{1}{c|}{48}                                                        & 112             & \multicolumn{1}{c|}{50}                                                         & \multicolumn{1}{c|}{68}            & \multicolumn{1}{c|}{48}                                                        & 105             \\ \hline
\textbf{800}                                                                           & \multicolumn{1}{c|}{70}                                                         & \multicolumn{1}{c|}{76}            & \multicolumn{1}{c|}{72}                                                        & 226             & \multicolumn{1}{c|}{70}                                                         & \multicolumn{1}{c|}{90}            & \multicolumn{1}{c|}{70}                                                        & 155             & \multicolumn{1}{c|}{67}                                                         & \multicolumn{1}{c|}{90}            & \multicolumn{1}{c|}{70}                                                        & 151             \\ \hline
\textbf{1000}                                                                          & \multicolumn{1}{c|}{96}                                                         & \multicolumn{1}{c|}{106}           & \multicolumn{1}{c|}{97}                                                        & 309             & \multicolumn{1}{c|}{95}                                                         & \multicolumn{1}{c|}{109}           & \multicolumn{1}{c|}{88}                                                        & 206             & \multicolumn{1}{c|}{90}                                                         & \multicolumn{1}{c|}{109}           & \multicolumn{1}{c|}{88}                                                        & 192             \\ \hline
\textbf{1200}                                                                          & \multicolumn{1}{c|}{118}                                                        & \multicolumn{1}{c|}{122}           & \multicolumn{1}{c|}{122}                                                       & 420             & \multicolumn{1}{c|}{115}                                                        & \multicolumn{1}{c|}{149}           & \multicolumn{1}{c|}{115}                                                       & 256             & \multicolumn{1}{c|}{113}                                                        & \multicolumn{1}{c|}{149}           & \multicolumn{1}{c|}{115}                                                       & 235             \\ \hline
\end{tabular}
\end{table}

\noindent
\textbf{Experiment 3:} 
\begin{table}[]
\begin{tabular}{|c|cccc|cccc|cccc|}
\hline
\multirow{2}{*}{\textbf{\begin{tabular}[c]{@{}c@{}}Dependency\\ percentage\end{tabular}}} & \multicolumn{4}{c|}{\textbf{CP1}}                                                                                                                                                                                       & \multicolumn{4}{c|}{\textbf{CP2}}                                                                                                                                                                                       & \multicolumn{4}{c|}{\textbf{CP3}}                                                                                                                                                                                       \\ \cline{2-13} 
                                                                                          & \multicolumn{1}{c|}{\textbf{\begin{tabular}[c]{@{}c@{}}ADJ\\ DAG\end{tabular}}} & \multicolumn{1}{c|}{\textbf{Tree}} & \multicolumn{1}{c|}{\textbf{\begin{tabular}[c]{@{}c@{}}LL\\ DAG\end{tabular}}} & \textbf{Serial} & \multicolumn{1}{c|}{\textbf{\begin{tabular}[c]{@{}c@{}}ADJ\\ DAG\end{tabular}}} & \multicolumn{1}{c|}{\textbf{Tree}} & \multicolumn{1}{c|}{\textbf{\begin{tabular}[c]{@{}c@{}}LL\\ DAG\end{tabular}}} & \textbf{Serial} & \multicolumn{1}{c|}{\textbf{\begin{tabular}[c]{@{}c@{}}ADJ\\ DAG\end{tabular}}} & \multicolumn{1}{c|}{\textbf{Tree}} & \multicolumn{1}{c|}{\textbf{\begin{tabular}[c]{@{}c@{}}LL\\ DAG\end{tabular}}} & \textbf{Serial} \\ \hline
\textbf{0}                                                                                & \multicolumn{1}{c|}{93}                                                         & \multicolumn{1}{c|}{87}            & \multicolumn{1}{c|}{82}                                                        & 398             & \multicolumn{1}{c|}{91}                                                         & \multicolumn{1}{c|}{99}            & \multicolumn{1}{c|}{84}                                                        & 404             & \multicolumn{1}{c|}{89}                                                         & \multicolumn{1}{c|}{90}            & \multicolumn{1}{c|}{81}                                                        & 164             \\ \hline
\textbf{20}                                                                               & \multicolumn{1}{c|}{94}                                                         & \multicolumn{1}{c|}{95}            & \multicolumn{1}{c|}{83}                                                        & 367             & \multicolumn{1}{c|}{92}                                                         & \multicolumn{1}{c|}{108}           & \multicolumn{1}{c|}{89}                                                        & 247             & \multicolumn{1}{c|}{94}                                                         & \multicolumn{1}{c|}{93}            & \multicolumn{1}{c|}{84}                                                        & 250             \\ \hline
\textbf{40}                                                                               & \multicolumn{1}{c|}{96}                                                         & \multicolumn{1}{c|}{98}            & \multicolumn{1}{c|}{86}                                                        & 329             & \multicolumn{1}{c|}{93}                                                         & \multicolumn{1}{c|}{112}           & \multicolumn{1}{c|}{94}                                                        & 213             & \multicolumn{1}{c|}{96}                                                         & \multicolumn{1}{c|}{97}            & \multicolumn{1}{c|}{90}                                                        & 301             \\ \hline
\textbf{60}                                                                               & \multicolumn{1}{c|}{98}                                                         & \multicolumn{1}{c|}{101}           & \multicolumn{1}{c|}{87}                                                        & 295             & \multicolumn{1}{c|}{94}                                                         & \multicolumn{1}{c|}{114}           & \multicolumn{1}{c|}{95}                                                        & 194             & \multicolumn{1}{c|}{98}                                                         & \multicolumn{1}{c|}{99}            & \multicolumn{1}{c|}{93}                                                        & 340             \\ \hline
\textbf{80}                                                                               & \multicolumn{1}{c|}{100}                                                        & \multicolumn{1}{c|}{102}           & \multicolumn{1}{c|}{91}                                                        & 244             & \multicolumn{1}{c|}{97}                                                         & \multicolumn{1}{c|}{132}           & \multicolumn{1}{c|}{98}                                                        & 191             & \multicolumn{1}{c|}{102}                                                        & \multicolumn{1}{c|}{105}           & \multicolumn{1}{c|}{96}                                                        & 386             \\ \hline
\textbf{100}                                                                              & \multicolumn{1}{c|}{103}                                                        & \multicolumn{1}{c|}{106}           & \multicolumn{1}{c|}{93}                                                        & 205             & \multicolumn{1}{c|}{103}                                                        & \multicolumn{1}{c|}{148}           & \multicolumn{1}{c|}{102}                                                       & 175             & \multicolumn{1}{c|}{105}                                                        & \multicolumn{1}{c|}{117}           & \multicolumn{1}{c|}{100}                                                       & 421             \\ \hline
\end{tabular}
\end{table}

\newpage
\subsection{Mixed Transaction Family}
\textbf{Experiment 1:}
\begin{table}[]
\begin{tabular}{|c|cccc|cccc|cccc|}
\hline
\multirow{2}{*}{\textbf{\begin{tabular}[c]{@{}c@{}}No of \\ Blocks\end{tabular}}} & \multicolumn{4}{c|}{\textbf{CP1}}                                                                                                                                                                                       & \multicolumn{4}{c|}{\textbf{CP2}}                                                                                                                                                                                       & \multicolumn{4}{c|}{\textbf{CP3}}                                                                                                                                                                                       \\ \cline{2-13} 
                                                                                  & \multicolumn{1}{c|}{\textbf{\begin{tabular}[c]{@{}c@{}}ADJ\\ DAG\end{tabular}}} & \multicolumn{1}{c|}{\textbf{Tree}} & \multicolumn{1}{c|}{\textbf{\begin{tabular}[c]{@{}c@{}}LL\\ DAG\end{tabular}}} & \textbf{Serial} & \multicolumn{1}{c|}{\textbf{\begin{tabular}[c]{@{}c@{}}ADJ\\ DAG\end{tabular}}} & \multicolumn{1}{c|}{\textbf{Tree}} & \multicolumn{1}{c|}{\textbf{\begin{tabular}[c]{@{}c@{}}LL\\ DAG\end{tabular}}} & \textbf{Serial} & \multicolumn{1}{c|}{\textbf{\begin{tabular}[c]{@{}c@{}}ADJ\\ DAG\end{tabular}}} & \multicolumn{1}{c|}{\textbf{Tree}} & \multicolumn{1}{c|}{\textbf{\begin{tabular}[c]{@{}c@{}}LL\\ DAG\end{tabular}}} & \textbf{Serial} \\ \hline
\textbf{10}                                                                       & \multicolumn{1}{c|}{118}                                                        & \multicolumn{1}{c|}{150}           & \multicolumn{1}{c|}{129}                                                       & 203             & \multicolumn{1}{c|}{136}                                                        & \multicolumn{1}{c|}{171}           & \multicolumn{1}{c|}{151}                                                       & 178             & \multicolumn{1}{c|}{134}                                                        & \multicolumn{1}{c|}{180}           & \multicolumn{1}{c|}{148}                                                       & 209             \\ \hline
\textbf{20}                                                                       & \multicolumn{1}{c|}{244}                                                        & \multicolumn{1}{c|}{308}           & \multicolumn{1}{c|}{269}                                                       & 400             & \multicolumn{1}{c|}{268}                                                        & \multicolumn{1}{c|}{347}           & \multicolumn{1}{c|}{300}                                                       & 355             & \multicolumn{1}{c|}{271}                                                        & \multicolumn{1}{c|}{356}           & \multicolumn{1}{c|}{293}                                                       & 433             \\ \hline
\textbf{30}                                                                       & \multicolumn{1}{c|}{356}                                                        & \multicolumn{1}{c|}{448}           & \multicolumn{1}{c|}{396}                                                       & 688             & \multicolumn{1}{c|}{396}                                                        & \multicolumn{1}{c|}{520}           & \multicolumn{1}{c|}{450}                                                       & 532             & \multicolumn{1}{c|}{406}                                                        & \multicolumn{1}{c|}{531}           & \multicolumn{1}{c|}{440}                                                       & 668             \\ \hline
\textbf{40}                                                                       & \multicolumn{1}{c|}{486}                                                        & \multicolumn{1}{c|}{604}           & \multicolumn{1}{c|}{531}                                                       & 879             & \multicolumn{1}{c|}{540}                                                        & \multicolumn{1}{c|}{682}           & \multicolumn{1}{c|}{603}                                                       & 727             & \multicolumn{1}{c|}{536}                                                        & \multicolumn{1}{c|}{723}           & \multicolumn{1}{c|}{585}                                                       & 893             \\ \hline
\textbf{50}                                                                       & \multicolumn{1}{c|}{594}                                                        & \multicolumn{1}{c|}{774}           & \multicolumn{1}{c|}{663}                                                       & 1094            & \multicolumn{1}{c|}{665}                                                        & \multicolumn{1}{c|}{876}           & \multicolumn{1}{c|}{751}                                                       & 893             & \multicolumn{1}{c|}{667}                                                        & \multicolumn{1}{c|}{891}           & \multicolumn{1}{c|}{736}                                                       & 1093            \\ \hline
\end{tabular}
\end{table}

\noindent
\textbf{Experiment 2:}
\begin{table}[]
\begin{tabular}{|c|cccc|cccc|cccc|}
\hline
\multirow{2}{*}{\textbf{\begin{tabular}[c]{@{}c@{}}Txns\\  per \\ Block\end{tabular}}} & \multicolumn{4}{c|}{\textbf{CP1}}                                                                                                                                                                                       & \multicolumn{4}{c|}{\textbf{CP2}}                                                                                                                                                                                       & \multicolumn{4}{c|}{\textbf{CP3}}                                                                                                                                                                                       \\ \cline{2-13} 
                                                                                       & \multicolumn{1}{c|}{\textbf{\begin{tabular}[c]{@{}c@{}}ADJ\\ DAG\end{tabular}}} & \multicolumn{1}{c|}{\textbf{Tree}} & \multicolumn{1}{c|}{\textbf{\begin{tabular}[c]{@{}c@{}}LL\\ DAG\end{tabular}}} & \textbf{Serial} & \multicolumn{1}{c|}{\textbf{\begin{tabular}[c]{@{}c@{}}ADJ\\ DAG\end{tabular}}} & \multicolumn{1}{c|}{\textbf{Tree}} & \multicolumn{1}{c|}{\textbf{\begin{tabular}[c]{@{}c@{}}LL\\ DAG\end{tabular}}} & \textbf{Serial} & \multicolumn{1}{c|}{\textbf{\begin{tabular}[c]{@{}c@{}}ADJ\\ DAG\end{tabular}}} & \multicolumn{1}{c|}{\textbf{Tree}} & \multicolumn{1}{c|}{\textbf{\begin{tabular}[c]{@{}c@{}}LL\\ DAG\end{tabular}}} & \textbf{Serial} \\ \hline
\textbf{200}                                                                           & \multicolumn{1}{c|}{44}                                                         & \multicolumn{1}{c|}{50}            & \multicolumn{1}{c|}{41}                                                        & 95              & \multicolumn{1}{c|}{49}                                                         & \multicolumn{1}{c|}{56}            & \multicolumn{1}{c|}{44}                                                        & 58              & \multicolumn{1}{c|}{50}                                                         & \multicolumn{1}{c|}{56}            & \multicolumn{1}{c|}{48}                                                        & 60              \\ \hline
\textbf{400}                                                                           & \multicolumn{1}{c|}{92}                                                         & \multicolumn{1}{c|}{107}           & \multicolumn{1}{c|}{80}                                                        & 143             & \multicolumn{1}{c|}{101}                                                        & \multicolumn{1}{c|}{120}           & \multicolumn{1}{c|}{94}                                                        & 125             & \multicolumn{1}{c|}{103}                                                        & \multicolumn{1}{c|}{120}           & \multicolumn{1}{c|}{100}                                                       & 124             \\ \hline
\textbf{600}                                                                           & \multicolumn{1}{c|}{136}                                                        & \multicolumn{1}{c|}{165}           & \multicolumn{1}{c|}{126}                                                       & 208             & \multicolumn{1}{c|}{156}                                                        & \multicolumn{1}{c|}{187}           & \multicolumn{1}{c|}{146}                                                       & 203             & \multicolumn{1}{c|}{154}                                                        & \multicolumn{1}{c|}{191}           & \multicolumn{1}{c|}{145}                                                       & 207             \\ \hline
\textbf{800}                                                                           & \multicolumn{1}{c|}{177}                                                        & \multicolumn{1}{c|}{233}           & \multicolumn{1}{c|}{177}                                                       & 259             & \multicolumn{1}{c|}{204}                                                        & \multicolumn{1}{c|}{262}           & \multicolumn{1}{c|}{204}                                                       & 265             & \multicolumn{1}{c|}{205}                                                        & \multicolumn{1}{c|}{268}           & \multicolumn{1}{c|}{205}                                                       & 269             \\ \hline
\textbf{1000}                                                                          & \multicolumn{1}{c|}{234}                                                        & \multicolumn{1}{c|}{307}           & \multicolumn{1}{c|}{235}                                                       & 349             & \multicolumn{1}{c|}{267}                                                        & \multicolumn{1}{c|}{341}           & \multicolumn{1}{c|}{276}                                                       & 350             & \multicolumn{1}{c|}{272}                                                        & \multicolumn{1}{c|}{353}           & \multicolumn{1}{c|}{279}                                                       & 354             \\ \hline
\textbf{1200}                                                                          & \multicolumn{1}{c|}{288}                                                        & \multicolumn{1}{c|}{380}           & \multicolumn{1}{c|}{320}                                                       & 490             & \multicolumn{1}{c|}{332}                                                        & \multicolumn{1}{c|}{441}           & \multicolumn{1}{c|}{384}                                                       & 441             & \multicolumn{1}{c|}{316}                                                        & \multicolumn{1}{c|}{439}           & \multicolumn{1}{c|}{355}                                                       & 434             \\ \hline
\end{tabular}
\end{table}

\noindent
\textbf{Experiment 3:} 
\begin{table}[]
\begin{tabular}{|c|cccc|cccc|cccc|}
\hline
\multirow{2}{*}{\textbf{\begin{tabular}[c]{@{}c@{}}Dependency\\ percentage\end{tabular}}} & \multicolumn{4}{c|}{\textbf{CP1}}                                                                                                                                                                                       & \multicolumn{4}{c|}{\textbf{CP2}}                                                                                                                                                                                       & \multicolumn{4}{c|}{\textbf{CP3}}                                                                                                                                                                                       \\ \cline{2-13} 
                                                                                          & \multicolumn{1}{c|}{\textbf{\begin{tabular}[c]{@{}c@{}}ADJ\\ DAG\end{tabular}}} & \multicolumn{1}{c|}{\textbf{Tree}} & \multicolumn{1}{c|}{\textbf{\begin{tabular}[c]{@{}c@{}}LL\\ DAG\end{tabular}}} & \textbf{Serial} & \multicolumn{1}{c|}{\textbf{\begin{tabular}[c]{@{}c@{}}ADJ\\ DAG\end{tabular}}} & \multicolumn{1}{c|}{\textbf{Tree}} & \multicolumn{1}{c|}{\textbf{\begin{tabular}[c]{@{}c@{}}LL\\ DAG\end{tabular}}} & \textbf{Serial} & \multicolumn{1}{c|}{\textbf{\begin{tabular}[c]{@{}c@{}}ADJ\\ DAG\end{tabular}}} & \multicolumn{1}{c|}{\textbf{Tree}} & \multicolumn{1}{c|}{\textbf{\begin{tabular}[c]{@{}c@{}}LL\\ DAG\end{tabular}}} & \textbf{Serial} \\ \hline
\textbf{0}                                                                                & \multicolumn{1}{c|}{281}                                                        & \multicolumn{1}{c|}{361}           & \multicolumn{1}{c|}{290}                                                       & 587             & \multicolumn{1}{c|}{277}                                                        & \multicolumn{1}{c|}{348}           & \multicolumn{1}{c|}{286}                                                       & 614             & \multicolumn{1}{c|}{277}                                                        & \multicolumn{1}{c|}{355}           & \multicolumn{1}{c|}{288}                                                       & 551             \\ \hline
\textbf{20}                                                                               & \multicolumn{1}{c|}{273}                                                        & \multicolumn{1}{c|}{371}           & \multicolumn{1}{c|}{284}                                                       & 602             & \multicolumn{1}{c|}{278}                                                        & \multicolumn{1}{c|}{388}           & \multicolumn{1}{c|}{296}                                                       & 563             & \multicolumn{1}{c|}{274}                                                        & \multicolumn{1}{c|}{411}           & \multicolumn{1}{c|}{276}                                                       & 542             \\ \hline
\textbf{40}                                                                               & \multicolumn{1}{c|}{280}                                                        & \multicolumn{1}{c|}{397}           & \multicolumn{1}{c|}{290}                                                       & 557             & \multicolumn{1}{c|}{277}                                                        & \multicolumn{1}{c|}{403}           & \multicolumn{1}{c|}{294}                                                       & 535             & \multicolumn{1}{c|}{277}                                                        & \multicolumn{1}{c|}{411}           & \multicolumn{1}{c|}{277}                                                       & 554             \\ \hline
\textbf{60}                                                                               & \multicolumn{1}{c|}{271}                                                        & \multicolumn{1}{c|}{412}           & \multicolumn{1}{c|}{296}                                                       & 550             & \multicolumn{1}{c|}{278}                                                        & \multicolumn{1}{c|}{407}           & \multicolumn{1}{c|}{294}                                                       & 516             & \multicolumn{1}{c|}{280}                                                        & \multicolumn{1}{c|}{404}           & \multicolumn{1}{c|}{280}                                                       & 542             \\ \hline
\textbf{80}                                                                               & \multicolumn{1}{c|}{275}                                                        & \multicolumn{1}{c|}{437}           & \multicolumn{1}{c|}{298}                                                       & 533             & \multicolumn{1}{c|}{273}                                                        & \multicolumn{1}{c|}{414}           & \multicolumn{1}{c|}{296}                                                       & 512             & \multicolumn{1}{c|}{276}                                                        & \multicolumn{1}{c|}{399}           & \multicolumn{1}{c|}{283}                                                       & 530             \\ \hline
\textbf{100}                                                                              & \multicolumn{1}{c|}{274}                                                        & \multicolumn{1}{c|}{444}           & \multicolumn{1}{c|}{293}                                                       & 476             & \multicolumn{1}{c|}{275}                                                        & \multicolumn{1}{c|}{423}           & \multicolumn{1}{c|}{315}                                                       & 468             & \multicolumn{1}{c|}{278}                                                        & \multicolumn{1}{c|}{464}           & \multicolumn{1}{c|}{320}                                                       & 533             \\ \hline
\end{tabular}
\end{table}

\newpage
\subsection{Smart Validation: Mixed Transaction Family}
\textbf{Experiment 1:}

\begin{table}[]
\begin{tabular}{|c|c|c|c|c|}
\hline
\textbf{No of Blocks} & \textbf{ADJ DAG} & \textbf{Tree} & \textbf{LL DAG} & \textbf{\begin{tabular}[c]{@{}c@{}}Smart\\ validation\end{tabular}} \\ \hline
\textbf{10}           & 0.80             & 0.27          & 0.37            & 0.10                                                                \\ \hline
\textbf{20}           & 1.62             & 0.59          & 0.65            & 0.31                                                                \\ \hline
\textbf{30}           & 2.31             & 0.79          & 1.00            & 0.42                                                                \\ \hline
\textbf{40}           & 2.37             & 0.86          & 1.50            & 0.65                                                                \\ \hline
\textbf{50}           & 2.95             & 1.25          & 1.70            & 0.65                                                                \\ \hline
\end{tabular}
\end{table}

\noindent
\textbf{Experiment 2:}
\begin{table}[]
\begin{tabular}{|c|c|c|c|c|}
\hline
\textbf{Txn per Block} & \textbf{ADJ DAG} & \textbf{Tree} & \textbf{LL DAG} & \textbf{\begin{tabular}[c]{@{}c@{}}Smart\\ validation\end{tabular}} \\ \hline
\textbf{200}           & 0.32             & 0.16          & 0.16            & 0.10                                                                \\ \hline
\textbf{400}           & 0.64             & 0.23          & 0.30            & 0.18                                                                \\ \hline
\textbf{600}           & 0.97             & 0.30          & 0.44            & 0.20                                                                \\ \hline
\textbf{800}           & 1.27             & 0.39          & 0.59            & 0.31                                                                \\ \hline
\textbf{1000}          & 1.60             & 0.50          & 0.76            & 0.29                                                                \\ \hline
\textbf{1200}          & 1.94             & 0.51          & 0.76            & 0.45                                                                \\ \hline
\end{tabular}
\end{table}

\noindent
\textbf{Experiment 3:}
\begin{table}[]
\begin{tabular}{|c|c|c|c|c|}
\hline
\textbf{\begin{tabular}[c]{@{}c@{}}Dependency\\ Percentage\end{tabular}} & \textbf{ADJ DAG} & \textbf{Tree} & \textbf{LL DAG} & \textbf{\begin{tabular}[c]{@{}c@{}}Smart\\ validation\end{tabular}} \\ \hline
\textbf{0}                                                               & 0.08             & 0.02          & 0.03            & 0.02                                                                \\ \hline
\textbf{20}                                                              & 0.07             & 0.02          & 0.03            & 0.02                                                                \\ \hline
\textbf{40}                                                              & 0.06             & 0.02          & 0.04            & 0.02                                                                \\ \hline
\textbf{60}                                                              & 0.06             & 0.02          & 0.04            & 0.02                                                                \\ \hline
\textbf{80}                                                              & 0.06             & 0.02          & 0.04            & 0.02                                                                \\ \hline
\textbf{100}                                                             & 0.06             & 0.03          & 0.05            & 0.02                                                                \\ \hline
\end{tabular}
\end{table}

\end{document}